\documentclass[a4paper,11pt]{article}
\usepackage{jheppub} 
\usepackage{subcaption}
\usepackage{mathrsfs}

\usepackage{physics}
\usepackage{color,xcolor}
\usepackage{bm}
\usepackage{mdframed}

\renewcommand\d{{\rm d}}

\usepackage{soul} 
\usepackage{amsthm}
\newtheorem{theorem}{Theorem}[section]

\newtheorem{lemma}[theorem]{Lemma}
\newtheorem{proposition}[theorem]{Proposition}
\newtheorem{conjecture}{Conjecture}[section]
\theoremstyle{definition}
\newtheorem{definition}{Definition}[section]
\theoremstyle{remark}
\newtheorem{remark}{Remark}

\title{\boldmath Horizon causality from holographic scattering in asymptotically dS${}_3$}

\author[a]{Victor Franken,}
\author[b,c]{Takato Mori}
\affiliation[a]{CPHT, CNRS, École polytechnique, Institut Polytechnique de Paris\\91120 Palaiseau, France}
\affiliation[b]{Perimeter Institute for Theoretical Physics, Waterloo, Ontario N2L 2Y5, Canada}
\affiliation[c]{
Center for Gravitational Physics and Quantum Information,
Yukawa Institute for Theoretical Physics, Kyoto University\\ 
Kitashirakawa Oiwakecho, Sakyo-ku, Kyoto 606-8502, Japan}

\emailAdd{victor.franken@polytechnique.edu, takato.mori@yukawa.kyoto-u.ac.jp}

\abstract{In the AdS/CFT correspondence, a direct scattering in the bulk may not have a local boundary analog. A nonlocal implementation on the boundary requires $O(1/G_N)$ mutual information. This statement is formalized by the connected wedge theorem, which can be proven using general relativity within AdS$_3$ but also argued for using quantum information theory on the boundary, suggesting that the theorem applies to any holographic duality. We examine scattering within the static patch of asymptotically dS$_3$ spacetime, which is conjectured to be described by a quantum theory on the stretched horizon in static patch holography. We show that causality on the horizon induced from null infinities~$\mathcal{I}^{\pm}$ is consistent with the theorem. Specifically, signals propagating in the static patch are associated with local operators at $\mathcal{I}^{\pm}$. Our results suggest a novel connection between static patch holography and the dS/CFT correspondence.}

\begin{document}

\begin{flushright}
YITP-24-133\\
CPHT-RR079.102024
\end{flushright}

\maketitle

\flushbottom

\section{Introduction}
\label{sec:intro}

The quantum mechanical description of the gravitational aspects of our universe, encompassing its geometry and causality, remains one of the most profound questions in theoretical physics. A lot of progress has been made via the explicit realization of the AdS/CFT correspondence~\cite{Maldacena:1997re,Witten:1998qj}, where anti-de Sitter (AdS) spacetime is dual to a conformal field theory (CFT) living on its asymptotic boundary. In parallel to these developments, there has been strong evidence that our universe has gone through two periods of accelerated expansion, which are the cosmic inflation in the early universe and the present period~\cite{SupernovaSearchTeam:1998fmf,SupernovaCosmologyProject:1998vns,Planck:2015fie}. An extensive understanding of the early universe asks for a quantum gravitational description, and extending the great progress made in AdS/CFT to expanding universes is one of the major challenges in modern high-energy physics. Toward a quantum mechanical description of our universe, which is approximately de Sitter (dS) spacetime, dS holography stands out as a promising but enigmatic framework~\cite{Strominger:2001pn,Bousso:2001mw,Abdalla:2002hg,Alishahiha:2004md,Parikh:2004wh,Alishahiha:2005dj,McFadden:2009fg,Dong:2010pm,Anninos:2011ui,Anninos:2011af,Dong:2018cuv,Gorbenko:2018oov,Arias:2019pzy, Arias:2019zug, Lewkowycz:2019xse,Geng:2021wcq,Susskind:2021esx,Shaghoulian:2021cef,Susskind:2021omt,Coleman:2021nor,Susskind:2021dfc,Hikida:2022ltr,Svesko:2022txo,Levine:2022wos,Shaghoulian:2022fop,Lin:2022nss, Susskind:2022dfz, Susskind:2022bia,  Banihashemi:2022htw,Rahman:2022jsf,Goel:2023svz, Narovlansky:2023lfz,Susskind:2023rxm, Giveon:2023rsk, Franken:2023pni, Kawamoto:2023nki, Galante:2023uyf}. One of its fundamental puzzles arises from the fact that dS spacetime is a closed universe, lacking a boundary in the traditional sense. This leads to the question: Where do the holographic degrees of freedom reside?

A natural extension of AdS/CFT to de Sitter spacetime is known as the dS/CFT correspondence~\cite{Strominger:2001pn, Bousso:2001mw, Balasubramanian:2001nb, Anninos:2011ui}, which is achieved via analytical continuation. This approach situates the dual CFT at future or past null infinity, where conventional notions of time and states are ill-defined. Consequently, standard concepts such as entanglement entropy, which are pivotal in understanding quantum gravity~\cite{Maldacena:2001kr,Ryu:2006bv,Hubeny:2007xt,VanRaamsdonk:2009ar,VanRaamsdonk:2010pw, Maldacena:2013xja,Dong:2016eik,Penington:2019npb,Almheiri:2019psf,Almheiri:2019hni, Doi:2023zaf}, become ambiguous.

An alternative proposal is static patch holography~\cite{Susskind:2021omt,Susskind:2021dfc,Susskind:2021esx, Shaghoulian:2021cef,Shaghoulian:2022fop,Lin:2022nss,Susskind:2023hnj}, in which the dual quantum theory is defined on a holographic screen, namely, a codimension-one timelike hypersurface. Under this proposal, the notion of time and unitarity is manifest. Here, a holographic screen is situated in the bulk, close to the cosmological horizon of an observer, and encodes the state of the static patch of this observer. The precise nature of the dual theory remains unsettled. While the double-scaled Sachdev-Ye-Kitaev (DSSYK) model is one conjecture~\cite{Susskind:2021esx,Lin:2022nss,Rahman:2022jsf,Goel:2023svz,Narovlansky:2023lfz,Verlinde:2024znh,Verlinde:2024zrh,Blommaert:2023opb, Blommaert:2023wad,Rahman:2024iiu,Milekhin:2023bjv,Xu:2024hoc,Milekhin:2024vbb}, it involves dimensional reduction, leaving the description of $d$-dimensional dS spacetimes for $d\ge 3$ unresolved. In particular, the dual theory in a higher-dimensional de Sitter spacetime is expected to be highly nonlocal, a feature that low-dimensional holography escapes. 
The ambiguity of the spatial extent of the holographic screen further complicates identifying the correct prescription for the holographic entanglement entropy, leading to various competing proposals~\cite{Susskind:2021esx,Shaghoulian:2021cef,Shaghoulian:2022fop, Franken:2023pni, Hao:2024nhd}.

To understand a holographic spacetime from a ultraviolet (UV) dual quantum theory, a natural approach is to identify a location for the dual degrees of freedom to live. In static patch holography, a holographic screen is placed away from the asymptotic boundary. Generally, a theory on such a screen is expected to be nonlocal so causality must be treated carefully.
Furthermore, the precise understanding of causality on the boundary/screen consistent with the bulk causality is essential to resolve various puzzles related to holographic entanglement entropy such as a violation of subadditivity~\cite{Kawamoto:2023nki,Mori:2023swn,Grado-White:2020wlb,Lewkowycz:2019xse}.

In this spirit, one interesting question is how a holographic quantum theory encodes the causal structure of the bulk dual~\cite{May:2019yxi,May:2019odp,May:2021nrl, May:2022clu,Caminiti:2024ctd}. 
Let us consider a situation where we send information from a set of input points, process it, and share the outcome among output points.
This procedure, known as a relativistic quantum task, can have a local bulk realization but not necessarily on the boundary. Whether the task can be performed locally in some background is purely a causality statement so it is entirely answered from the causal structure and locations of input and output points.
The connected wedge theorem~\cite{May:2019odp} in AdS$_3$/CFT$_2$ relates this causality statement to correlation, namely, that when a $2$-to-$2$ scattering from input and output points on the boundary is possible in the bulk, but not on the boundary, there must be $O(1/G_N)$ mutual information between certain boundary causal diamonds, which are fixed by the input and output locations and the boundary causality.

The connected wedge theorem in AdS$_3$/CFT$_2$ has been shown both from the bulk, using general relativity~\cite{May:2019odp}, and from the boundary using quantum information without relying on the detailed nature of the bulk or boundary theory~\cite{May:2019yxi,May:2019odp}. This suggests that this theorem should be valid beyond the AdS/CFT correspondence. In this paper, we consider holographic scattering in the context of static patch holography. In particular, we seek to clarify how the connected wedge theorem may be realized in de Sitter space, and draw lessons on causality on holographic screen and a potential connection between static patch holography and the dS/CFT correspondence.

This paper is organized as follows. For readers only interested in our main claim, skip to Section~\ref{sec:results} for a summary of the main results or Section~\ref{sec:causality} for further details. 
In Section~\ref{sec:assump-summary}, we provide the assumptions we make about the spacetimes in the paper and review our main results.
In Section~\ref{sec:recap}, we provide a short recap on the connected wedge theorem in asymptotically AdS and one with branes or cutoff surfaces~\cite{Mori:2023swn}. Section~\ref{sec:SPH} briefly reviews the geometry of dS and static patch holography.
In Section~\ref{sec:causality}, we present a puzzle where the connected wedge theorem is violated in static patch holography. We then resolve this problem by introducing the notion of an induced lightcone, which is constructed from a point at null infinity. Additionally, we examine how induced causality changes as we push the holographic screen deep inside the static patch.
In Section~\ref{sec:theorem}, we first show how induced causality resolves the apparent violation of the connected wedge theorem by explicitly calculating holographic entanglement entropy. We then prove the theorem in the static patch of asymptotically dS$_3$ spacetime using induced causality on the holographic screen. 
Possible future directions are discussed in Section~\ref{sec:discussion}.

In Appendix~\ref{app:notations}, we list the notation of symbols we use throughout the paper.
In Appendix~\ref{app:ext_proc}, we discuss formal aspects related to covariant holographic entanglement entropy prescriptions in static patch holography. In particular, we review the definitions of the extremization and maximin procedure as well as the constrained extremization proposed in~\cite{Franken:2023pni}, and prove their equivalence in the static patch.

\section{Main results and notations}\label{sec:assump-summary}

In this section, we briefly review the assumptions about the spacetimes, and present a concise summary of our main results.

\subsection{Assumptions about the spacetimes}
\label{sec:assump}

We assume that the spacetimes discussed in this paper are classical solutions of Einstein's equations; we mainly consider smooth, asymptotically (A)dS$_3$ spacetime. We always assume (AdS) global hyperbolicity, which states that the spacetime, or its conformal compactification in the case of AdS, has a Cauchy slice.
Additionally, unless stated otherwise, we assume that these spacetimes satisfy the null energy condition
\begin{equation}
\label{eq:NEC}
    k^{\mu}k^{\nu}T_{\mu\nu} \geq 0,
\end{equation}
for any null vector $k^{\nu}$, where $T_{\mu\nu}$ is the stress tensor. On some occasions, we will mention the extension of some results in spacetimes that violate the null energy condition. In such cases, we consider the semiclassical limit $G_N\rightarrow 0$ and the usual theorems of general relativity that rely on the null energy condition, such as the focusing theorem~\cite{Wald:1984rg} or the second law of causal horizons~\cite{Jacobson:2003wv}, must be replaced by their conjectured quantum versions, such as the restricted quantum focusing conjecture~\cite{Bousso:2015mna, Shahbazi-Moghaddam:2022hbw} or generalized second law for causal horizons~\cite{Wall:2009wm,Wall:2011hj}. These allow us to replace the classical null energy condition \eqref{eq:NEC} with the quantum null energy condition~\cite{Bousso:2015mna, Ceyhan:2018zfg, Bousso:2015wca, Balakrishnan:2017bjg, Koeller:2015qmn} which is a fundamental statement about quantum field theory. 

Later we will introduce a holographic screen near the cosmological horizon to discuss static patch holography. Since our scope in this paper is scattering by probe particles, we fix the background geometry and place a screen at a fixed location. Thus, in this paper we are agnostic about the precise boundary condition imposed on the screen although if we consider any perturbation, we need to identify a sensible boundary condition with stability, diffeomorphism, etc. We refer to~\cite{Anninos:2024wpy} for a study on boundary conditions in static patch and~\cite{Karch:2013oqa} for a possible boundary condition on the screen.

We treat areas as finite quantities and stay quite lax concerning quantities such as the expansion. Indeed, null hypersurfaces may present cusps where null generators of the surface collide such that the expansion diverges. We assume that they can be dealt with and neglect them, and refer to Section~$3$ of~\cite{May:2019odp} for a rigorous approach to these cusps. 

\subsection{Brief summary of main results}
\label{sec:results}

In this paper, we present a puzzle in pure dS$_3$ where the connected wedge theorem is violated. We consider a $2$-to-$2$ scattering with input and output points on the cosmological horizon of an observer, where the holographic screen is supposed to be located by the static patch holography proposal.\footnote{Sometimes we place a Planck-distance cutoff, replacing the screen with the stretched horizon. Even in such a case, the same puzzle occurs and the general proof presented in the later section will apply to both cases.} 
For simplicity, we choose these points such that the scattering in the bulk barely occurs -- the scattering region is pointlike. See Figure~\ref{fig:scat-scr} for a schematic picture. 
\begin{figure}
    \centering
    \includegraphics[width=0.8\linewidth]{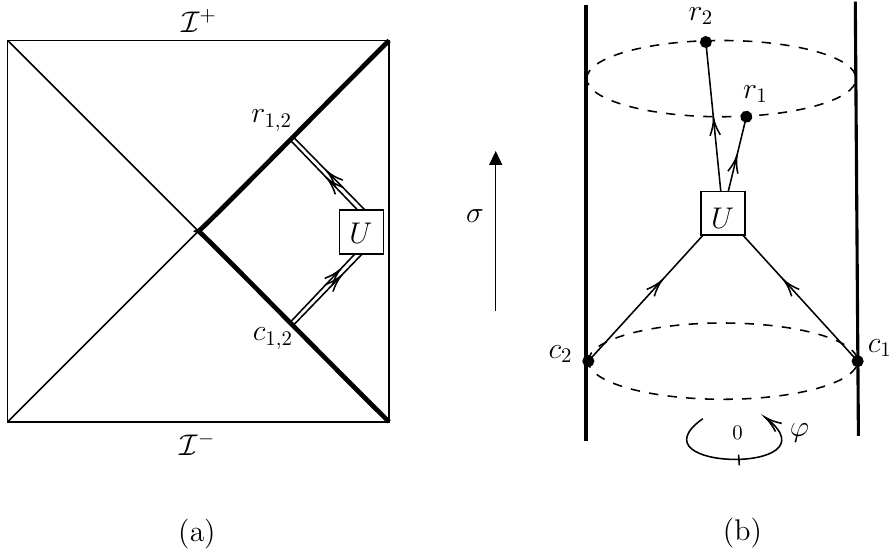}
    \caption{A schematic picture of {a} 2-to-2 scattering in the static patch {of pure dS${}_3$}, from input points $c_1,c_2$ to output points $r_1,r_2$ on the cosmological horizon. 
    The white box denoted $U$ corresponds to a unitary operation acting on the signal qubits in the bulk. In both pictures, $\sigma$ {indicates} the time direction. (a) Picture of the scattering in the Penrose diagram of de Sitter space. Each slice of the diagram has the topology of a sphere and the vertical lines on the edges follow the trajectories of static observers on the poles of the sphere. The diagonal lines are the cosmological horizons of the two observers. (b) Picture of the scattering in the static patch. 
    Each cosmological horizon (or stretched horizon) has the topology of a cylinder and the bulk (static patch) fills inside.
    Input and output points are chosen such that a direct scattering is only possible at one bulk point. The angular coordinate on the cylinder is denoted by $\varphi$.}
    \label{fig:scat-scr}
\end{figure}
By tuning the spatial location of the input and output points, one can create a situation where the $2$-to-$2$ scattering is impossible on the screen, at least from local interactions. The connected wedge theorem predicts that a quantum task via scattering is non-locally implemented on the screen by $O(1/G_N)$ entanglement between two causally disconnected screen subsystems. However, the highly constrained causality due to the lightlike nature of the holographic screen leads us to conclude that these subsystems are too small to accommodate the expected large entanglement based on the holographic entanglement entropy proposal for static patch holography~\cite{Shaghoulian:2021cef, Franken:2023pni}.

We argue that this apparent violation of the connected wedge theorem is due to an inconsistent method for computing causal regions on the holographic screens. In the original version of the connected wedge theorem~\cite{May:2019odp}, it is taken for granted that a localized wavepacket propagating through the bulk from the boundary is created by a local operator on the boundary~\cite{Terashima:2023mcr}. However, in general, this is a non-trivial statement. We claim that the apparent failure of the connected wedge theorem stems back to this assumption regarding the screen causality of a bulk signal. 

To resolve this puzzle, we need to understand which screen causality is being used.
We distinguish three different causalities on the screen. 1) Causality based on the induced metric of the screen $\mathcal{S}$ (which will be denoted by a subscript ${}_{\mathcal{S}}$), 2) causality determined from the intersection of the bulk lightcone of a point on the screen and the screen itself (which will be denoted by conditioning $(\cdot)\vert_{\mathcal{S}}$), and 3) induced causality, defined as an intersection between the bulk lightcone from a null infinity and the screen (which will be denoted by an accent $\hat{\phantom{a}}$). The first type and the second type of causalities may look similar, however, they are different in general. In Section~\ref{sec:causality}, we present a case where they are indeed different.

We propose that the relevant causality is different from the one based on the induced metric and/or on a local theory. There are several reasons to think holographic theory on the stretched horizon is nonlocal. These arguments include the fact that this theory displays hyperfast scrambling~\cite{Susskind:2021esx}, and the volume-law of entanglement entropy on the horizon~\cite{Shaghoulian:2022fop}. As nonlocality often allows superluminal signaling, it seems sufficient for the resolution of the puzzle. However, we seek a finer resolution; what kind of (apparent) nonlocality or superluminality is needed for the connected wedge theorem, and what is its origin?

Inspired by a recent work by one of the authors~\cite{Mori:2023swn}, we propose the consistent resolution of the puzzle is given by the causality induced by a local operator from the null infinities. We then show the induced causality as defined in Section~\ref{sec:causality} resolves the apparent violation of the connected wedge theorem as the induced lightcone (which should not be confused with a lightcone constructed from the induced metric) extends farther than the local lightcones, leading to larger causal regions associated with the scattering.
This leads us to state a general version of the connected wedge theorem in the static patch of three-dimensional asymptotically de Sitter spacetime.\footnote{In this paper, we often use the term static patch to describe the causal region of an observer, by analogy with the case of the pure de Sitter spacetime. Thus, the geometry of this bulk region is not necessarily described by a static metric.} Before stating the theorem, let us first define the induced causality:

\begin{definition}[Induced causality]\label{def:induced}
    Let $\mathcal{S}$ be the holographic screen and $J^+(\tilde{c})$ be the future bulk lightcone from a point $\tilde{c}$ on $\mathcal{I}^-$. Given an input point $c$ on $\mathcal{S}$, pick an input point $\tilde{c}$ on $\mathcal{I}^-$ such that $c \in J^+(\tilde{c})$.
    Then, the future induced lightcone of $c$ on $\mathcal{S}$ is defined as
    \begin{equation}
        \hat{J}^+(c) = \min_{\tilde{c}\,\in\, \mathcal{I}^-} [J^+(\tilde{c})\cap \mathcal{S}],
    \end{equation}
    where the minimization is taken so that the {minimizing} point $\tilde{c}=\tilde{c}_\ast$ satisfies $ J^+(\tilde{c}_\ast)\cap \mathcal{S}\subseteq  J^+(\tilde{c}^\prime)\cap \mathcal{S}$ for any $\tilde{c}^\prime$. Analogously, we define a past induced lightcone by replacing $c$ with $r$ and $+$ with $-$. See Figure~\ref{fig:ind_lc} for a visualization.
\end{definition}
The existence of the minimizing points $\tilde{c},\tilde{r}$ is guaranteed from the definition of null infinities $\mathcal{I}^\pm$. Now, let us state the connected wedge theorem in asymptotically dS${}_3$:
\begin{figure}
    \centering

        \includegraphics[width=8cm]{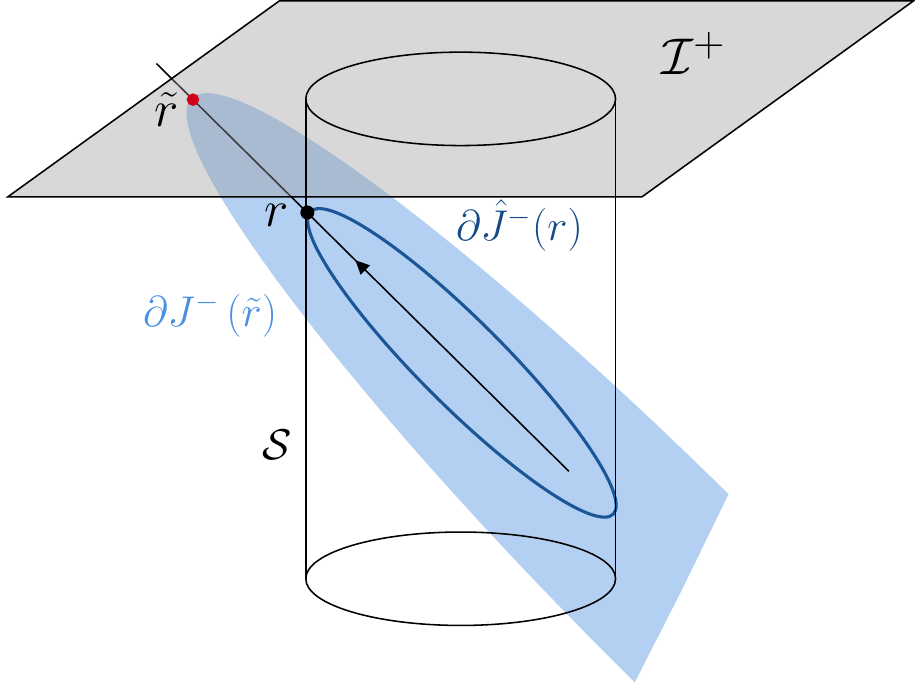}
    \caption{A bulk perturbation propagating towards the boundary is pictured by an arrow. Given an output $r_1$ on the holographic screen $\mathcal{S}$ (cylinder here), a point $\tilde r$ is defined on the future conformal boundary $\mathcal{I}^+$ (in grey). The boundary of the past induced lightcone of $r$ (in dark blue) is given by the intersection of the past lightcone of $\tilde{r}$ (light blue) with the screen. Similarly, by exchanging $+$ and $-$, future and past, one obtains a future induced lightcone.}
            \label{fig:ind_lc}
\end{figure}

\begin{theorem}[Connected wedge theorem in static patch]
 \label{th:CWTpatchsimplified}
Let $\mathcal{S}$ be the holographic screen in the static patch of an asymptotically dS$_3$ spacetime. Assuming static patch holography and its associated entropy prescription, if the $2$-to-$2$ scattering $c_1,c_2 \in \mathcal{S} \rightarrow r_1,r_2\in\mathcal{S}$ is possible in the bulk but not on the screen $\mathcal{S}$ based on the induced causality, then {the decision regions} $\mathcal{R}_1{\equiv\hat{J}^+(c_1)\cap \hat{J}^-(r_1) \cap \hat{J}^-(r_2)}$ and $\mathcal{R}_2{\equiv \hat{J}^+(c_2)\cap \hat{J}^-(r_1) \cap \hat{J}^-(r_2)}$
have mutual information as large as $O(1/G_N)$, implying their entanglement wedge is connected in the bulk bounded by $\mathcal{S}$. 
\end{theorem}
A precise version of this theorem is given in Theorem~\ref{th:CWTpatch}. This statement is proven using the notion of induced causality on the holographic screen. The proof is analogous to that of~\cite{May:2019odp, Mori:2023swn}, using the extremality and {its} equivalence to the maximin prescription. Note that our proof of the connected wedge theorem applies not only to scattering among points on the cosmological horizon but also on a general holographic screen, defined in Definition~\ref{def:ap-sp} and Definition~\ref{def:ap}.

The induced causality on the holographic screen from asymptotic boundaries $\mathcal{I}^{\pm}$ hints at a possible relation between static patch holography and the dS/CFT correspondence. In particular, causality on the screen consistent with holographic entanglement entropy in static patch holography is induced from boundary local operators in dS/CFT.

Finally, the appendix of this paper contains formal results regarding the subtleties in the definition of holographic entropy prescription(s) in static patch holography. Identifying the prescription is of the highest importance in the context of the connected wedge theorem, as the proof relies on the extremality of the surface computing entropy as well as the equivalence of the surface to the maximin one. However, this is not trivial in the proposed dS holographic entropy prescriptions.
So far two proposals have been made: the monolayer and bilayer proposals~\cite{Susskind:2021esx,Shaghoulian:2021cef,Shaghoulian:2022fop}. The monolayer proposal appears inconsistent with the entanglement wedge reconstruction as pointed out in a paper by one of the authors~\cite{Franken:2023pni}. However, even if we adopt the bilayer proposal, there remains an ambiguity regarding the `extremization'. In contrast to AdS/CFT, where an extremized surface always exists and is equivalent to the maximin surface~\cite{Wall:2012uf}, the extremization problem may have no solution in static patch holography due to the non-asymptotic feature of the holographic screen. This can be resolved by relaxing the global extremizing condition, leading to the so-called constrained extremization~\cite{Franken:2023pni}. However, there remains a subtlety as it is not guaranteed equivalent to the maximin surface as in AdS/CFT. Indeed, we prove that they are different in general (Theorem~\ref{ineq}). Nevertheless, we find these three definitions (extremization, constrained extremization, and maximin) are equivalent within the static patch (Theorem~\ref{th:patch}).

\section{Connected wedge theorem in AdS$_3$/CFT$_2$}
\label{sec:recap}

The AdS/CFT correspondence is a duality between the asymptotically anti-de Sitter (AdS) bulk and a conformal field theory (CFT) on the boundary. In the semiclassical limit, this duality suggests that a quantity defined in the holographic CFT can be computed by a geometric quantity in the bulk and vice versa.
For instance, the entanglement entropy of a state in a holographic CFT can be calculated as the area of the associated extremal surface. 
This duality also applies to the dynamics, revealing an illuminating interplay between quantum operations and causality. In this section, we review one particular application, known as the connected wedge theorem~\cite{May:2019yxi}. 

\subsection{Holographic scattering from the asymptotic boundary}\label{sec:asympt-task}
\begin{figure}
    \centering
    \includegraphics[width=0.8\linewidth]{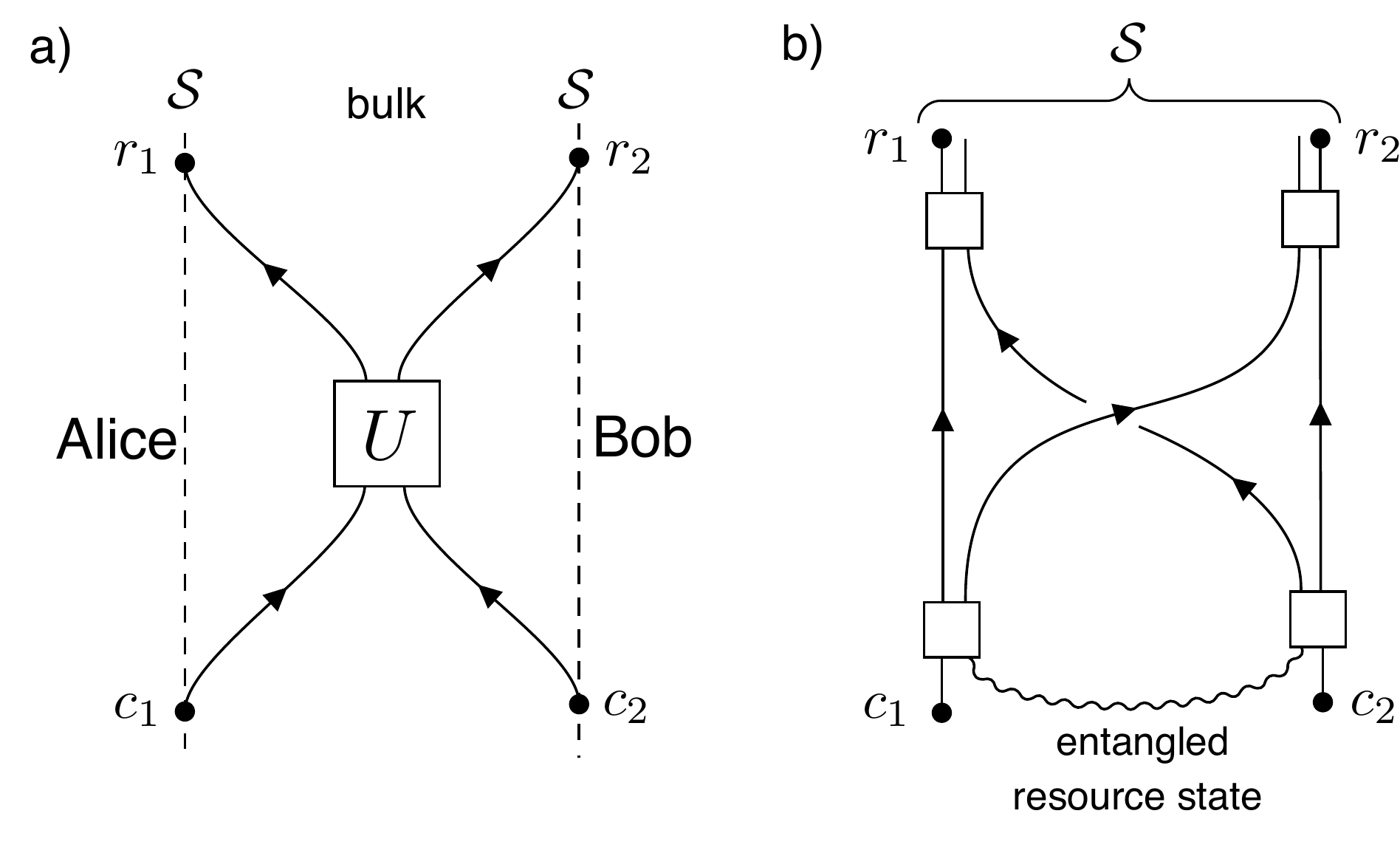}
    \caption{Two realizations of a 
    quantum task. a) A local implementation: 
    Applying a unitary gate $U$ between Alice and Bob's qubits via {a} direct scattering in the bulk. b) An equivalent, nonlocal realization: the task can be achieved on the boundary/screen without direct contact using an entangled resource state between Alice and Bob. White boxes represent some local operations. This is known as a nonlocal quantum computation.}
    \label{fig:NLQC}
\end{figure}

In this subsection, we consider an asymptotic quantum task via holographic scattering {in AdS$_3$/CFT$_2$.  
Let Alice and Bob each have a qubit. They input each qubit at their respective input locations $c_1,c_2$, then process them through a holographic evolution, and receive their qubits back again at their respective output locations $r_1,r_2$. For the connected wedge theorem, we first consider performing a certain quantum task in the bulk through a direct scattering. The input and output points are placed on the boundary so that the causal future of the input points and the causal past of the output points have an intersection in the bulk, which is the scattering region. We configure the holographic evolution such that a unitary gate acts locally on the qubits in the scattering region to accomplish the task.} This protocol accomplishes the quantum task through {a} direct scattering as shown in Figure~\ref{fig:NLQC}a.\footnote{{Note that in principle, any unitary process can be implemented in a holographic setting.}} 
Now, let us turn to the boundary perspective. Since we work in holography, the task must be also achievable on the boundary, however, there are cases where no direct scattering is possible on the boundary.
Thus, we need to find a way to perform the same task nonlocally. Such a task is known as a nonlocal quantum computation~\cite{PhysRevLett.90.010402, Buhrman_2014}. In general, entanglement between Alice and Bob is required to accomplish the task nonlocally (Figure~\ref{fig:NLQC}b).
By considering a particular task called the \textbf{B}${}_{\bm{84}}$, it can be proven that to accomplish the task there must be a finite correlation between Alice and Bob. The relevant correlation must be available after the input is received and in the past of both outputs (Figure~\ref{fig:NLQC}b). This implies that there must be a finite mutual information between each decision region $\mathcal{R}_{1,2}$, defined as the intersection of the future domain of dependence of each input and the past domains of dependence of two outputs. The mutual information of these two regions is defined as
\begin{equation}
\label{eq:mut_inf}
    I(\mathcal{R}_1:\mathcal{R}_2)=S(\mathcal{R}_1) + S(\mathcal{R}_2) - S(\mathcal{R}_1\cup\mathcal{R}_2),
\end{equation}
where $S(\mathcal{R})$ is the entanglement entropy between $\mathcal{R}$ and its complement. Note that the definition of the decision regions heavily relies on the causal structure on the boundary $\mathcal{S}$. To accomplish the \textbf{B}${}_{\bm{84}}$ task $n$ times in parallel with high probability, one can show that the mutual information must be $I(\mathcal{R}_1:\mathcal{R}_2)\gtrsim n$. 
By taking $n$ any function as large as $o(1/G_N)$ to avoid backreaction, one can argue that the mutual information must be $O(1/G_N)$~\cite{May:2021nrl}.\footnote{The notation $f(x)=o(g(x))$ as $x\rightarrow \infty$ means that for any constant $\epsilon>0$ there exist $x_0$ such that $|f(x)|\leq \epsilon g(x)$ for $x\geq x_0$.} Note that this does not rely on the large-$N$ separation between the area term and the quantum correction of the mutual information.\footnote{This point is important, since later the dual quantum theory remains unknown so holographic entanglement entropy may not have a nice split into the area term of order $1/G_N$ (or possibly higher order) and quantum corrections of order unity.
There is possibly a term growing with a rate that is sub-linear but larger than any $1/G_N^a$ with $a<1$. Even if this is the case, the argument in the main body says the mutual information must be {at least as large as $1/G_N$,} implying a geometric correlation from the area terms. We thank Alex May for pointing out a loophole in the first version of this paper, and explaining this to us.}

By taking $n$ as large as $O(1/G_N^a)$ with $a<1$ to avoid backreaction, one can argue that the mutual information must be $O(1/G_N)$ due to the large-$N$ hierarchy between the area term and the quantum correction. In summary, the connected wedge theorem is given as follows:
\begin{theorem}[Connected wedge theorem~\cite{May:2019yxi}]\label{th:CWT}
    Consider two input points $c_1,c_2$ and output points $r_1,r_2$ on the boundary of an asymptotically AdS$_3$ spacetime. If a 2-to-2 scattering $c_1,c_2\rightarrow r_1,r_2$ is possible in the bulk, \emph{i.e.}
    \begin{equation}
        J_{12\rightarrow 12} = J^+(c_1)\cap J^+(c_2)\cap J^-(r_1) \cap J^-(r_2) \neq \varnothing,
    \end{equation}
    but not on the asymptotic boundary $\mathcal{S}$, \emph{i.e.}
    \begin{equation}
        J^+(c_1)\vert_\mathcal{S}\cap J^+(c_2)\vert_\mathcal{S}\cap J^-(r_1)\vert_\mathcal{S} \cap J^-(r_2)\vert_\mathcal{S} = \varnothing,
    \end{equation}
    then the decision regions $\mathcal{R}_{1,2}=J^+(c_{1,2})\vert_\mathcal{S}\cap J^-(r_1)\vert_\mathcal{S} \cap J^-(r_2)\vert_\mathcal{S}$ on the boundary must have a connected entanglement wedge, {equivalently}
    \begin{equation}
        I(\mathcal{R}_1:\mathcal{R}_2)=O(1/G_N).
    \end{equation}
\end{theorem}

This discussion is based on quantum information and does not rely on the details of the theory.\footnote{One caveat is that we ignore the complementary regions of the decision regions. As pointed out in~\cite{May:2019odp} and the erratum of~\cite{May:2019yxi}, GHZ-type resource states among the complementary regions and $\mathcal{R}_1\cup\mathcal{R}_2$ can offer a protocol that does not need a finite mutual information between $\mathcal{R}_1$ and $\mathcal{R}_2$. This is however unlikely the case for holographic scattering as we do not expect a large amount of the GHZ-type entanglement in holography~\cite{Susskind:2014yaa,Nezami:2016zni,Mori:2024gwe}.}
Thus, it should work in any holographic spacetime. In this paper, we use this connected wedge theorem as a guiding principle to constrain or check various proposals related to the holographic duality. One application in the previous study will be reviewed in the next subsection.

We note that for the asymptotic quantum task in AdS$_3$/CFT$_2$, there is a gravitational proof based on the focusing theorem and quantum extremal surface formula of the holographic entanglement entropy. This is explained in~\cite{May:2019odp} (see also~\cite{Mori:2023swn} for a short review). We will follow the basic strategy later for the dS case. We also note that the connected wedge theorem was initially claimed to be valid in any dimensions~\cite{May:2019odp,May:2019yxi,May:2021nrl}. It was however pointed out in~\cite{May:2022clu} that the geometric and quantum information arguments are not valid above three bulk dimensions. For the same reasons, the results of this paper are also strictly restricted to dS$_3$.

\subsection{Holographic scattering from a non-asymptotic boundary}

While the original proposal of the connected wedge theorem considers an asymptotic quantum task, where the nonlocal quantum computation takes place on the asymptotic boundary, there is no reason not to consider more general holographic spacetime such as non-AdS and/or a holographic screen not located at the asymptotic boundary.

The authors of~\cite{Mori:2023swn}, including one of us, extended the connected wedge theorem to a braneworld or cutoff surface in an asymptotically AdS$_3$ spacetime. In these setups, the holographic screen $\mathcal{S}$ is located somewhere other than the asymptotic boundary. It turns out that the causality based on the induced metric on $\mathcal{S}$ leads to an apparent violation of the connected wedge theorem. The resolution presented in the work is to fill behind the hypersurface with a fictitious asymptotically AdS space,\footnote{In the work~\cite{Mori:2023swn}, the focus was a braneworld/cutoff AdS so just extending the original spacetime beyond the hypersurface was sufficient to resolve the puzzle. In general, one can glue an arbitrary fictitious spacetime with a fictitious asymptotic boundary as long as it satisfies the Israel junction condition~\cite{Israel:1966rt} to be a smooth spacetime satisfying the Einstein equation.} and extend the scattering trajectories to the fictitious asymptotic boundary to define fictitious input and output points $\tilde{c}_{1,2}\in J^-(c_{1,2})$, $\tilde{r}_{1,2}\in J^+(r_{1,2})$. See Figure~\ref{fig:brane-ind} for its illustration.
The authors have shown that the boundary domains of dependence defined from the induced causality, denoted by $\hat{J}^\pm$, align with the connected wedge theorem.
After all,~\cite{Mori:2023swn} proposes the following refined connected wedge theorem:
\begin{theorem}[Refined connected wedge theorem]\label{thm:refined-CWT}
    Two input points $c_1,c_2$ and output points $r_1,r_2$ are on a holographic screen $\mathcal{S}$, which is not necessarily the conformal boundary of asymptotically AdS$_3$ spacetime. Suppose a 2-to-2 scattering $c_1,c_2\rightarrow r_1,r_2$ is possible in the bulk, \emph{i.e.}
    \begin{equation}
        J_{12\rightarrow 12} = J^+(c_1)\cap J^+(c_2)\cap J^-(r_1) \cap J^-(r_2) \neq \varnothing,
    \end{equation}
    but not on the holographic screen $\mathcal{S}$, \emph{i.e.}
    \begin{equation}
        \qty[J^+(\tilde{c}_1)\cap J^+(\tilde{c}_2)\cap J^-(\tilde{r}_1)\cap J^-(\tilde{r}_2)] \cap \mathcal{S} = \varnothing,
    \end{equation}
    where the boundary causality determining $\hat{J}^\pm (p)$ is given by the induced lightcones $\hat{J}^{\pm}(p)=J^\pm (\tilde{p})\cap\mathcal{S}$ from a fictitious point $\tilde{p}$ on the fictitious asymptotic boundary.\footnote{See also Definition~\ref{def:induced}.}
    Then, the decision regions $\mathcal{R}_{1,2}=\hat{J}^+(c_{1,2})\cap \hat{J}^-(r_1) \cap \hat{J}^-(r_2)$ on the screen $\mathcal{S}$ must have a connected entanglement wedge, implied from
    \begin{equation}
        I(\mathcal{R}_1:\mathcal{R}_2)=O(1/G_N).
    \end{equation}
\end{theorem}

While the fictitious spacetime and boundary behind the holographic screen are not necessarily unique, the induced causality from a local point on the fictitious boundary is anticipated from the apparent nonlocality/superluminality of the boundary theory based on holographic renormalization group flow~\cite{Freedman:1999gp,Girardello:1998pd,Distler:1998gb,McGough:2016lol}. The fictitious boundary behind the holographic screen serves as the `true' UV boundary, and a fictitious local excitation on the UV boundary induces an effective, apparently nonlocal excitation dual to a localized signal in the bulk. This idea of the induced lightcone identifies causal and entanglement structures consistent with the holographic description.
The connected wedge theorem serves as a nontrivial check for the induced lightcone proposal.

\begin{figure}
    \centering
    \includegraphics[width=0.3\linewidth]{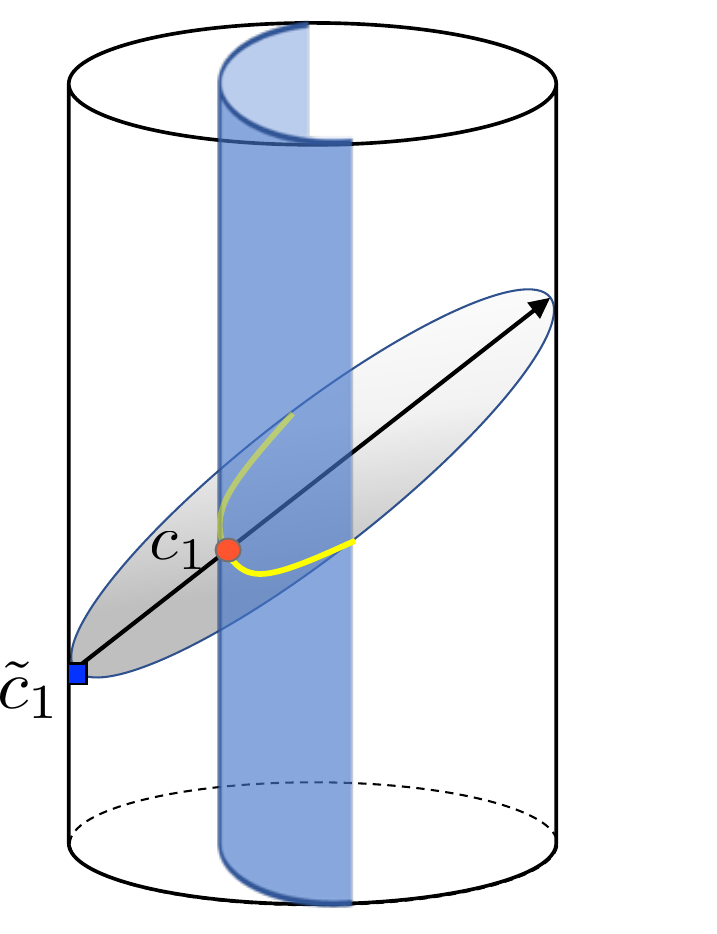}
    \caption{Given an input point $c_1$ on the holographic screen $\mathcal{S}$ (blue surface), a fictitious input point $\tilde{c}_1$ is defined on the fictitious asymptotic boundary. The intersection of the bulk lightcone emanating from $\tilde{c}_1$ and the screen $\mathcal{S}$ defines the induced lightcone on $\mathcal{S}$.}
    \label{fig:brane-ind}
\end{figure}

We emphasize that this induced lightcone approach in the light of the connected wedge theorem amounts to identifying the UV boundary that describes a non-local boundary theory locally. This determines the causality associated with holographic entanglement on the holographic screen. 

Unlike the cases considered in~\cite{Mori:2023swn}, where holographic entanglement entropy prescription is examined in some explicit models on the brane or a cutoff surface, so far there is no explicit model of static patch holography for dS$_3$. Hence the holographic entanglement entropy prescription is not verified from the dual quantum theory. However, by using the fact that various proposals reduce to a single prescription in a static patch (Appendix~\ref{app:proofs}), we assume the prescription and ask what causality on the screen is consistent with the dS holography. As the connected wedge theorem has its origin in quantum information, it should be true regardless of the dual UV theory. Thus, the theorem offers a nontrivial criterion on causality and a consistency check with the holographic entanglement entropy prescription.

\section{Static patch holography}
\label{sec:SPH}

The simplest model of the universe undergoing accelerated expansion is given by the dS geometry. We focus on three dimensions as there exist subtleties related to the proof of the connected wedge theorem in higher dimensions~\cite{May:2022clu}.

In conformal coordinates, the metric takes the form
\begin{equation}
\label{eq:conf}
    \d s^2 = \frac{1}{\cos^2\sigma}\left(-\d \sigma^2 +\d\theta^2 +\sin^2\theta\d\varphi^2 \right),
\end{equation}
where $\sigma\in(-\pi/2,\pi/2)$, $\theta\in[0,\pi]$, and $\varphi\in[0,2\pi)$. {See Figure~\ref{fig:Penrose_dS} for the associated Penrose diagram.}
\begin{figure}
    \centering
    \includegraphics[width=0.4\linewidth]{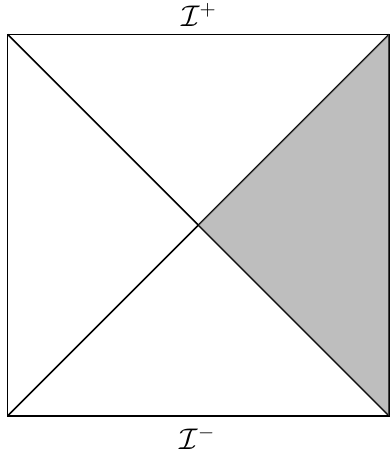}
    \caption{{Penrose diagram of de Sitter space. Two horizontal lines are the past and future null infinities $\mathcal{I}^{\pm}$. The shaded region is the causal region -- static patch -- of a static observer located at $\theta=\pi$ (south pole).}}
    \label{fig:Penrose_dS}
\end{figure}
Here the radius of curvature $l_{dS}$ has been set to $1$. Future and past null infinities $\mathcal{I}^\pm$ are located at $\sigma=\pm\pi/2$, respectively.
Once an observer is defined as a causal worldline in spacetime, de Sitter spacetime has the particularity to possess observer-dependent horizons. Indeed, the region where an observer can send and receive signals never covers the whole universe. Such a region is often called the static patch -- or causal region -- of the observer and it is bounded by a cosmological horizon.\footnote{As mentioned in footnote 2, the static patch in this paper is not necessarily described by a time-independent metric.} 

Contrary to black hole horizons, cosmological horizons depend on the worldline of the observer, as they are constructed as the union of the causal past and causal future of the future and past endpoints of the worldline. {The observer's and antipodal observer's worldlines are given by $\theta=\pi$ (south pole) and $\theta=0$ (north pole) in the conformal coordinates, respectively.} The {cosmological} horizon can be made manifest by writing the dS metric in static coordinates,
\begin{equation}
\label{eq:stat}
    \d s^2=-(1-r^2)\d t^2 + \frac{\d r^2}{1-r^2}+r^2\d\varphi^2,
\end{equation}
where $t\in\mathbb{R}$ and $r\in[0,1]$. The worldline of the observer is located at $r\rightarrow 0$, and its cosmological horizon is at $r=1$. Note that this coordinate system only covers the static patch of the observer. In particular, there is no global future-directed timelike Killing vector in dS spacetime.

An alternative definition of dS space, which will be very useful in the following, is given as a hypersurface embedded in the Minkowski spacetime
\begin{align}
    \d s^2 &= -dX_0^2+dX_3^2+dX_1^2+dX_2^2.
\end{align}
dS spacetime is then defined by considering the induced metric from the constraint
\begin{equation}
\label{eq:dS_hyperb}
   1 = -X_0^2 +X_3^2 +X_1^2 +X_2^2.
\end{equation}
One can show that the metrics \eqref{eq:conf} and \eqref{eq:stat} satisfy this constraint by parametrizing the embedding coordinates as
\begin{alignat}{2}
        X_0 & =  \tan \sigma &&= \sqrt{1-r^2} \sinh t, \\
        X_3 & =  \frac{\cos\theta}{\cos\sigma} &&= - \sqrt{1-r^2}\cosh t, \\
        X_1 & =  \frac{\sin\theta}{\cos\sigma} \cos\varphi && = r \cos\varphi, \\
        X_2 & =  \frac{\sin\theta}{\cos\sigma} \sin\varphi && = r \sin\varphi.
\end{alignat}
Note that we matched the time direction between the conformal patch and the right static patch.
The static coordinates above only cover the right static patch as $X_3\ge 0$ in the static coordinates above, which corresponds to $\theta\in[\pi/2,\pi]$.

Holographically describing dS spacetime is conceptually complicated. One of the main reasons for this is the absence of timelike boundaries. Indeed, spacelike slices of de Sitter are topologically closed; they do not have any boundary. One proposal by Strominger is that the holographic dual is located at null infinity $\mathcal{I}^+$, a conformal spacelike boundary of the spacetime~\cite{Strominger:2001pn, Bousso:2001mw}. This can be seen as a sort of a Euclidean continuation of AdS/CFT, known as the \emph{dS/CFT correspondence}. However, one of the largest differences between dS/CFT and AdS/CFT is that the notion of time is lacking in the former. Thus, the dual CFT is considered to be exotic like non-unitary and/or having an imaginary central charge~\cite{Hikida:2022ltr}. 
Even when the dS spacetime itself is Lorentzian, the dual CFT defined on $S^d$ at future null infinity $\mathcal{I}^+$ (which is a conformal boundary) is Euclidean. 

Identifying the Gibbons-Hawking dS entropy as a counting of holographic degrees of freedom, we have another proposal, called \emph{static patch holography}~\cite{Susskind:2021omt}. This conjectures that a quantum theory located on a stretched horizon encodes the state of its interior. In this paper, we refer to the interior of a cosmological horizon as the side containing the observer's worldline and the exterior as the side not contained in either static patch. 
We mainly focus on static patch holography in this paper; however, brief comments will be made later in relation to the other proposal.

Following~\cite{Susskind:2021omt, Franken:2023pni}, let us make the static patch holography proposal more precise. Consider an arbitrary observer in asymptotically dS spacetime. Achronal slices of the bulk are denoted by $\Sigma$. 
To encode bulk dynamics holographically, a holographic screen should be a timelike, codimension-one hypersurface. Furthermore, we expect that the (code subspace) algebra of the holographic screen should be equal to that of the observer's worldline $\zeta$. Since the timelike tube theorem~\cite{Borchers1961,osti_4665531,Strohmaier:2023opz} states that the algebra of the observer's worldline is equivalent to that of operators in its timelike envelope $E(\zeta)$~\cite{Witten:2023xze}, the domain of dependence of the holographic screen should equal $E(\zeta)$. This means $\partial\zeta$ should coincide with the future and past timelike edge of the holographic screen.\footnote{Physically, this means no signal from the observer travels beyond the holographic screen without crossing it. Conversely, any signal sent from a point beyond the screen cannot be received by the observer without crossing the screen.}
Additionally, the holographic screen should be convex to avoid any subtleties regarding holographic entanglement entropy prescription and induced causality~\cite{Mori:2023swn}.\footnote{This condition could possibly be removed by considering a generalized entanglement wedge~\cite{Bousso:2022hlz}.} After these considerations, we define the holographic screen in static patch holography as follows:
\begin{definition}[Holographic screen in static patch of asymptotically dS]
\label{def:ap-sp}
    The holographic screen $\mathcal{S}$ associated with an observer is a codimension-one convex timelike hypersurface in the static patch, anchored to the observer's worldline endpoints.
\end{definition}
\begin{remark}
    We can generalize this notion of the holographic screen to more general spacetimes such as a Friedmann-Lemaître-Robertson-Walker spacetime by explicitly requiring non-positive expansion toward the observer's worldline to ensure the Bousso bound~\cite{Bousso:1999xy} as follows:
\end{remark}
\begin{definition}[Holographic screen]
    \label{def:ap}
    The holographic screen $\mathcal{S}$ associated with an observer is the codimension-one convex timelike boundary of a region in which all closed codimension-two surfaces for which null geodesics orthogonal to the surface and directed towards the observer's worldline are of non-positive expansion. When referring to a holographic screen without mentioning an observer, $\mathcal{S}$ may be the union of holographic screens associated with different observers.
\end{definition}

In pure dS, $\mathcal{S}$ is any convex timelike codimension-one surface inside the cosmological horizon. 
When it is located at a distance of order $O(G_N)$ away from the horizon, it is called a stretched horizon. Note that our definition naturally generalizes to holographic screens associated with finite-lifetime observers.\footnote{We thank Cynthia Keeler for discussions on this point.}

\begin{conjecture}[Static patch holography~\cite{Susskind:2021omt}]
\label{conj:SPH}
The semiclassical gravity description of the region inside the screen $\mathcal{S}$ is dual to a quantum theory defined on $\mathcal{S}$.\footnote{Based on earlier works~\cite{Susskind:2021esx,Shaghoulian:2021cef, Shaghoulian:2022fop}, it was argued by the authors of~\cite{Franken:2023pni}, including one of us, that two stretched horizons associated with two antipodal observers can not only encode their associated interior but also the bulk region separating them. See Appendix~\ref{sec:FPRT} for more details.} 
\end{conjecture}

This conjecture was initially motivated by the covariant entropy bound~\cite{Susskind:2021omt, Bousso:1999xy, Bousso:2002ju} and then developed in~\cite{Susskind:2021dfc, Susskind:2021esx, Lin:2022nss, Susskind:2023hnj}. See~\cite{Shyam:2021ciy, Lewkowycz:2019xse, Coleman:2021nor, Banihashemi:2022htw, Banihashemi:2022jys} for additional supportive evidence. Moreover, a recent proposal of Conjecture~\ref{conj:SPH} for asymptotically dS${}_3$ identifies $\mathcal{S}$ with either the worldline of the two antipodal observers or the two stretched horizons~\cite{Susskind:2021esx,Lin:2022nss,Rahman:2022jsf,Goel:2023svz,Narovlansky:2023lfz,Verlinde:2024znh,Verlinde:2024zrh,Blommaert:2023opb,Blommaert:2023wad,Rahman:2024iiu,Milekhin:2023bjv,Xu:2024hoc,Milekhin:2024vbb}. Following a dimensional reduction, the dual theory is conjectured to be the double-scaled Sachdev-Ye-Kitaev (DSSYK) model, with a number of nontrivial matches, including correlation functions, the partition function, quasinormal modes, hyperfast scrambling and the algebra of observables. Conjecture~\ref{conj:SPH} will be one of the main assumptions in our proof of the connected wedge theorem in the static patch (Theorem~\ref{th:CWTpatch}).

To compute holographic entanglement entropy of subsystems of spatial slices of the holographic screen, we will consider the bilayer proposal~\cite{Shaghoulian:2021cef,Franken:2023pni}. A detailed introduction to this prescription is provided in Appendix~\ref{app:ext_proc}. Consider a spacelike subsystem of the screen, $A\in \mathcal{S}\vert_{\Sigma}$. Its holographic entanglement entropy in the proposal is given by the sum of the areas of the two extremal surfaces homologous to $A$ on both sides of the screen, $D$ and $D_E$, where $D$ denotes the codimension-zero interior of the screen, and $D_E$ denotes the codimension-zero region sandwiched by the screen and the complementary one associated with an antipodal observer. For their precise definition, see Appendix~\ref{app:ext_proc} and Fig.~\ref{fig:penrose}. (Note that $D$ is denoted as $D_R$ in the appendix.) The holographic entanglement entropy is proposed as
\begin{equation}
    S(A) = \frac{1}{4G_N}\left[\mathrm{area}(\gamma_e(A;D))+\mathrm{area}(\gamma_e(A;D_E))\right]+O(1),
\end{equation}
where $\gamma_e(A;D)$ and $\gamma_e(A;D_E)$ are the minimal extremal surfaces homologous to $A$ in $D$ and $D_E$, respectively.

We show in Appendix~\ref{app:proofs} that analogous to the AdS case, $\gamma_e(A;D)$ can be found by a maximin procedure, which amounts to finding the minimal area surface on any spatial slice $\Sigma$ such that $\partial\Sigma=\mathcal{S}\vert_{\Sigma}$, and choosing the maximum-area surface among them:
\begin{equation}
    \mathrm{area}(\gamma_e(A;D)) = \max_{\Sigma}\left(\min_{\gamma\in\Sigma}\left(\mathrm{area}(\gamma)\right)\right).
\end{equation}

In contrast, $\gamma_e(A;D_E)$ depends on which extremal surface prescription we use, as mentioned in Appendix~\ref{app:def}. In particular, extremal surfaces are not always maximin surfaces in $D_E$, and conversely. Nevertheless, it does not matter for the discussion of the connected wedge theorem as explained now. The homology condition in $D_E$ suggests that any subsystem of the screen is \emph{not} homologous to its complementary regions on the screen, due to the presence of another boundary. Thus, given two spacelike subsystems $A$ and $B$ of the screen,
$\gamma_e(A\cup B;D_E)$ is necessarily equal to $\gamma_e(A;D_E)\cup \gamma_e(A;D_E)$. It follows that the contributions from $D_E$ cancel out in the mutual information. Finally, the mutual information is purely given in terms of the area in the interior as
\begin{equation}
\label{eq:for}
    I(A:B)=\frac{1}{4G_N}\left[\mathrm{area}(\gamma_e(A;D))+\mathrm{area}(\gamma_e(B;D))-\mathrm{area}(\gamma_e(A\cup B;D))\right].
\end{equation}

\section{Causality on the holographic screen}
\label{sec:causality}

In this section, we present an apparent failure of the connected wedge theorem in dS$_3$ spacetime. As the theorem has a proof based on quantum information, which should apply in any holographic setup, this poses a puzzle. {(Note that the quantum information argument implicitly requires unitarity. See a remark after this paragraph.)} We resolve it by revisiting causality on the holographic screen. This resolution leads to three important consequences: 1) An insight for the dS connected wedge theorem, presented in Section~\ref{sec:theorem} for three-dimensional asymptotically dS spacetime. 2) Bulk local excitations in the interior of an observer's holographic screen, emanating from the screen, are not described by local operators on the screen. 3) A local excitation on $\mathcal{S}$ should be mapped to a local operator on $\mathcal{I}^{\pm}$, hinting at a relation between static patch holography and the dS/CFT correspondence.

{
\begin{remark}
    In Section~\ref{sec:asympt-task}, we implicitly assume unitarity on the screen/boundary in the quantum information argument for the connected wedge theorem. One may think this is not the case for a static patch of dS as it is a part of the global dS. We argue however that this would be the case. Recall that there is a timelike Killing vector in a static patch of pure dS. This implies that there is a Hamiltonian that generates a time translation so the state in the static patch is expected to evolve unitarily, according to the Liouville-von Neumann equation. This does not contradict with the static patch being a part of the global spacetime since the state can be mixed; in fact, it should be an (almost) maximally mixed state carrying dS entropy. We should stress that, however, a static patch may not have a timelike Killing vector in asymptotically dS, as mentioned in footnote 2. Nevertheless, our gravitational proof that we present later in Section~\ref{sec:proof} works even in these cases by defining the static patch as the causal region (timelike envelope) of the observer's worldline.
\end{remark}
}

\subsection{An apparent violation of the connected wedge theorem}
\label{sec:violation}

Let us consider the limiting case where the holographic screen $\mathcal{S}$ of an observer in pure dS${}_3$ is located on the cosmological horizon. We consider here an example of a $2$-to-$2$ scattering $c_1,c_2\rightarrow r_1,r_2$ with
\begin{align}
\label{eq:Pscat}
\begin{alignedat}{2}
    c_1 &= (-\pi/4,3\pi/4,\pi/2), &\quad c_2 &= (-\pi/4,3\pi/4,-\pi/2), \\
    r_1 &= (\pi/4,3\pi/4,0), &\quad r_2 &= (\pi/4,3\pi/4,\pi),
\end{alignedat}
\end{align}
in conformal coordinates $(\sigma,\theta,\varphi)$. This scattering is possible in the bulk. In particular, the scattering can only occur at one point:
\begin{equation}
    J_{12\rightarrow 12} = J^+(c_1)\cap J^+(c_2) \cap J^-(r_1) \cap J^-(r_2) = (0,\pi/2,0).
\end{equation}
We would like to test if the connected wedge theorem (Theorem~\ref{th:CWT}) applies to this example of scattering in dS spacetime. For this, one computes the decision regions, where each decision region $\mathcal{R}_i$ is defined as an intersection on the screen $\mathcal{S}$ among the future domain of dependence of $c_i$ and the past domains of dependence of $r_1$ and $r_2$ for $i=1,2$.

To this end, one needs to identify the appropriate causality on the screen $\mathcal{S}$. 
One possibility would be the causality based on the induced metric on $\mathcal{S}$. In this case, the screen causality is completely blind to the holographic bulk. In the current limiting case, the discrepancy between the bulk causality and boundary causality is the highest, as $\mathcal{S}$ is located on the cosmological horizon, i.e. the induced metric is
\begin{equation}
    ds^2_\mathcal{S}=d\varphi^2.
\end{equation}
Thus, under this causality, only light can propagate at a fixed angle $\varphi$. This suggests that the decision region is empty as the spatial location of each input and output point differs.\footnote{There is a subtlety at $\sigma=0$ as the cosmological horizon bifurcates. However, this subtlety can be avoided without significantly modifying the metric by replacing $\mathcal{S}$ with the stretched horizon. Thus, the conclusion should be the same.} Note that the decision region becomes nonempty only when the spatial location of one of the input points and both output points coincide. Even if this happens, the decision region is pointlike so it leads to a significant violation of the connected wedge theorem. Indeed, equation \eqref{eq:for} leads to a vanishing mutual information since the area of $\gamma_e(A;D)$ goes to $0$ as $A$ becomes pointlike. The disconnected extremal surfaces $\gamma_e(A;D)\cup \gamma_e(B;D)$ therefore always dominate when $A$ and $B$ are pointlike. 
See Appendix~\ref{app:ext_proc} for more details on the holographic entanglement entropy prescription.

Another possibility is to consider the bulk causality restricted on the screen. In this case, each decision region is given by $J^+(c_{i})\vert_\mathcal{S}\cap J^-(r_1)\vert_\mathcal{S} \cap J^-(r_2)\vert_\mathcal{S}$ for $i=1,2$, where we define the lightcone of a point $p$ on the screen as
\begin{equation}
\label{eq:deflc}
    J^{\pm}\vert_{\mathcal{S}}(p) = J^{\pm}(p) \cap \mathcal{S},\quad p\in\mathcal{S}.
\end{equation}
To find the decision region, we need to find the intersection between the boundary of the bulk lightcone of a point on the horizon and the cosmological horizon itself. This reduces to finding the points in the embedding spacetime that are lightlike separated from the horizon, and on the horizon. Noting that the horizon in the embedding coordinates are $(X_0=T>0,X_3=-T,X_1=X,X_2=\pm\sqrt{1-X^2})$, the intersection is found by solving
\begin{equation}
    \begin{split}
        1+TX_0-XX_1\pm \sqrt{1-X^2}X_2-TX_3 &=0,\\
        X_3 &= - |X_0|,\\
        X_1^2+X_2^2&=1.
    \end{split}
\end{equation}
The solution in conformal coordinates is
\begin{equation}
    \begin{cases}
        \begin{alignedat}{2}
            \varphi &= \varphi_0 &\quad &(\sigma\sigma_0 \geq 0)\\
            \tan\sigma &= -\cot\sigma_0\sin^2\left(\frac{\varphi-\varphi_0}{2}\right) &\quad &(\sigma\sigma_0 \leq 0)
        \end{alignedat}
    \end{cases}
    ,
    \label{eq:embed-sol}
\end{equation}
where $(\sigma_0,\varphi_0)$ are the coordinates on the horizon of $p$. An example lightcone on the horizon for $T=\sigma_0> 0$ is pictured in Figure~\ref{fig:LC}.
\begin{figure}[h]
    \centering
    \begin{subfigure}[t]{0.3\textwidth}
    \centering
    \includegraphics[width=5cm]{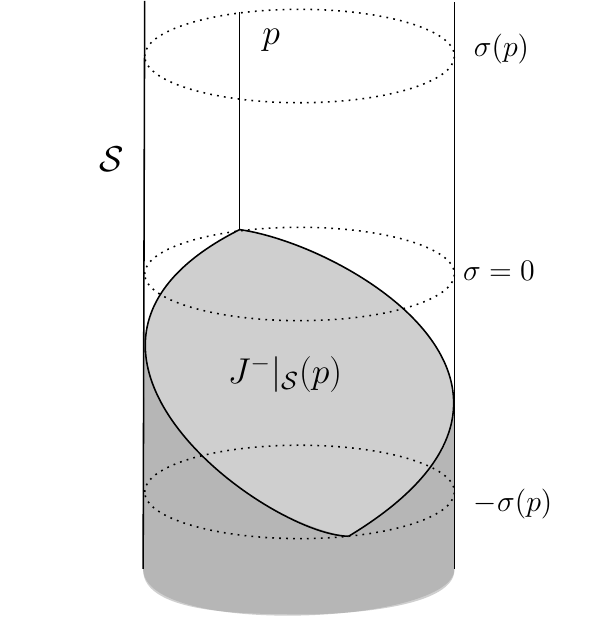}
    \caption{}
    \label{fig:LC}
    \end{subfigure} 
    \hfill
    \begin{subfigure}[t]{0.6\textwidth}
    \centering
    \includegraphics[width=10cm]{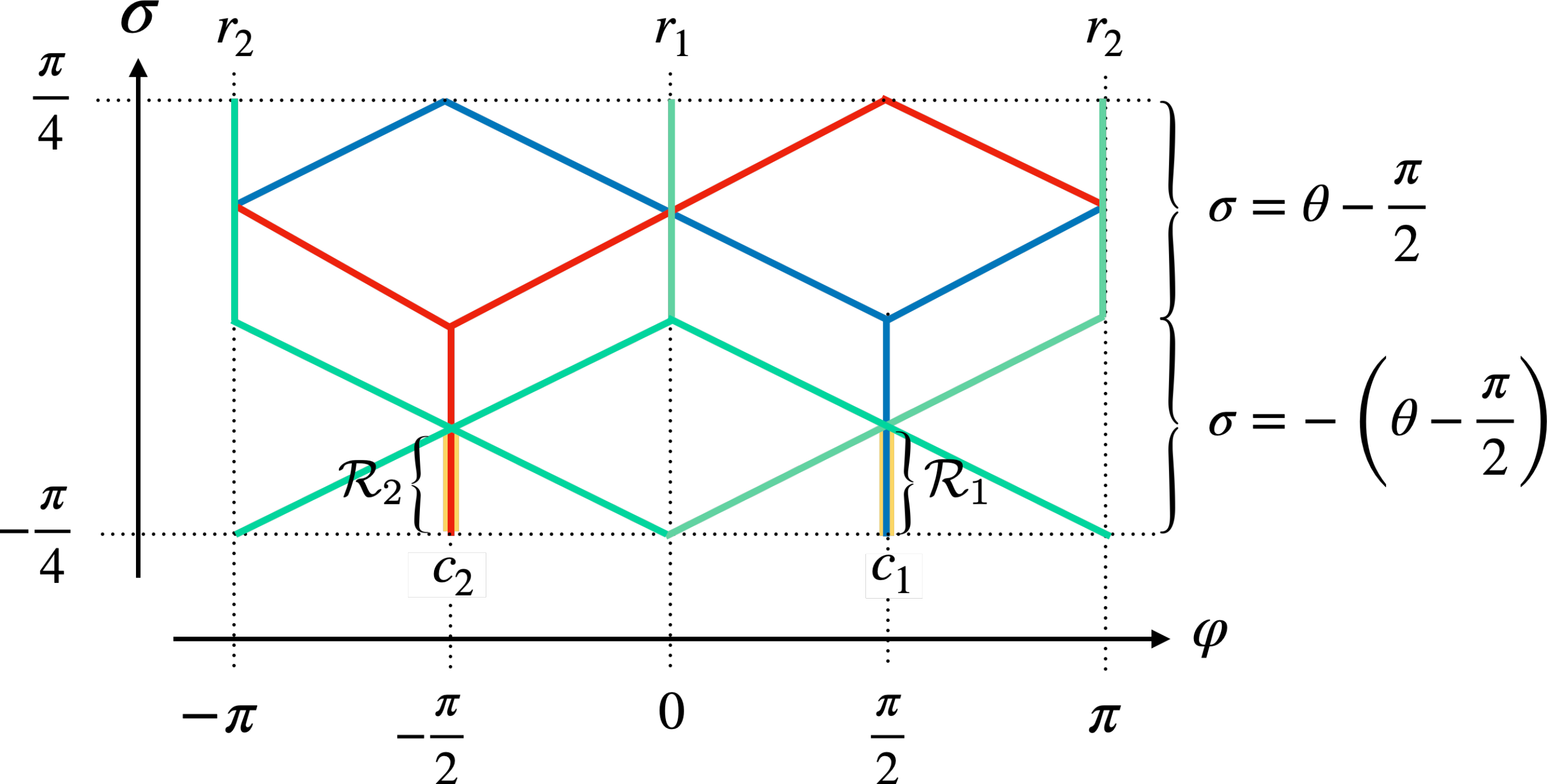}
    \caption{}
    \label{fig:decR}
    \end{subfigure}
    \caption{(a) Past lightcone (shaded in grey) of a point $p$ on the screen $\mathcal{S}$. The screen is located on the cosmological horizon, which has the topology of a cylinder $X_1^2+X_2^2=1$. The boundary of the lightcone consists of the union of the segment in the region $\sigma > 0$ with the ellipsoid curve in the region $\sigma \leq  0$. (b) Decision regions $\mathcal{R}_1$ and $\mathcal{R}_2$ for the scattering \eqref{eq:Pscat} constructed from lightcones on the screen $\mathcal{S}$. $\partial J\vert_{\mathcal{S}}^+(c_1)$ is denoted by a red line, $\partial J\vert_{\mathcal{S}}^+(c_2)$ is denoted by a blue line, and $\partial J\vert_{\mathcal{S}}^-(r_{1,2})$ are denoted by green lines. The decision regions have pointlike spatial sections, and therefore vanishing entanglement.}
\end{figure}
Using equation \eqref{eq:embed-sol}, we find that the decision regions associated with the input and output points~\eqref{eq:Pscat} are two timelike segments (Figure~\ref{fig:decR}). If the connected wedge theorem is true in dS, the decision regions should be connected by an entanglement wedge, but this is not the case here. A spacelike slice of the decision regions reduces to a set of two points on the horizon. The holographic entanglement entropy of two pointlike regions is zero and the associated entanglement wedge vanishes. This is an explicit apparent counterexample to the connected wedge theorem in dS spacetime.
This presents a puzzle as the connected wedge theorem (Theorem~\ref{th:CWT}) is expected to hold in any holographic setup.

\subsection{Induced causality from dS boundary \texorpdfstring{$\mathcal{I}^{\pm}$}{} }
\label{sec:ind_LC}
{What needs to be modified for
the connected wedge theorem to hold?} Since we start from the semiclassical limit of the holographic duality, the existence of the bulk scattering from the bulk causality should be taken for granted.
On the other hand, there is room to change the causality and decision regions fixed from it on the screen where the dual theory lives. We implicitly assumed through the definition of the decision regions that a bulk signal emanating from or reaching a point on the screen is given by an excitation of a local operator on the screen. 
This is the case in AdS/CFT, where a bulk excitation can be created by a local operator on the asymptotic boundary~\cite{Nozaki:2013wia,Terashima:2023mcr}. It was noted by one of the authors in~\cite{Mori:2023swn} that this is not the case in AdS holography with braneworld or cutoff surfaces. 
They suggest that a correct boundary dual of a localized wave packet in the bulk is given by a local excitation on a fictitious boundary so that it is effectively smeared and nonlocal on the brane/cutoff surface.\footnote{Quantum mechanically, the fictitious evolution is interpreted as a state preparation on the hypersurface by a finite Lorentzian time evolution. Alternatively, a fine-tuned causal structure from microscopic aspects of the dual theory may resolve the puzzle~\cite{Omiya:2021olc}. In this paper, we seek a resolution that only relies on the geometry, without such a fine tuning of the theory.}
This hints that we need to define another notion of causality on the non-asymptotic boundary, namely, the induced causality.

In this work, we follow the strategy of~\cite{Mori:2023swn} and consider `fictitious' local perturbations at the conformal boundaries of asymptotically de Sitter spacetime. We quote the word `fictitious' because in dS the screen is not a boundary, where spacetime terminates, so the conformal boundary is not fictitious contrary to the previous work in AdS.
This will be first motivated by the analogy with the case studied in~\cite{Mori:2023swn}, reviewed in Section~\ref{sec:recap}, and related to geometric properties of causal horizons, which will be made precise in Section~\ref{sec:proof}. The second motivation is the dS/CFT correspondence~\cite{Strominger:2001pn}, which would provide a physical interpretation of the fictitious points providing the induced lightcones of perturbations on the holographic screen. 

Considering a $2$-to-$2$ scattering among points on the screen, $c_1,c_2,\in \mathcal{S} \rightarrow r_1,r_2 \in \mathcal{S}$, we define `fictitious' input and output points as follow.
\begin{definition}[`Fictitious'/tilded point]
\label{def:R}
    We define points on the conformal boundaries $\Tilde{c}_1,\Tilde{c}_2\in \mathcal{I}^-,\Tilde{r}_1,\Tilde{r}_2\in \mathcal{I}^+$ as points that are causally connected to the points $c_1,c_2,r_1,r_2$ on the screen $\mathcal{S}$. That is,
    \begin{equation}
        \tilde{c}_i \in J^-(c_i) \cap \mathcal{I}^- \quad,\quad \tilde{r}_i \in J^+(r_i) \cap \mathcal{I}^+.
    \end{equation}
\end{definition}
The tilded point associated with a point on the screen leads to the definition of an induced lightcone:
\begin{definition}[Induced causality]
\label{def:ind}
    The induced lightcone of a point $p\in\mathcal{S}$ induced from $\tilde p$ is
    \begin{equation}
        \hat J^{\pm}(p) = J^{\pm}(\tilde p) \cap \mathcal{S}.
    \end{equation}
\end{definition}
See Figure~\ref{fig:ind_lc} for a schematic example of the induced past lightcone. For convenience, we define the induced lightcone in this way, however, due to its arbitrariness of choosing $\tilde{p}$, there are multiple alternative induced lightcones for a given point $p\in\mathcal{S}$. Most of these induced lightcones may look counterintuitive as these future/past induced lightcones can contain a point causally in the past/future of $p$ (in terms of the bulk causality), respectively. In general, $\hat{J}^{\pm}(p)$ may refer to any of them arbitrarily, and the proof of the connected wedge theorem is valid for any choice of $\tilde{p}$, as we will see in the next section. However, in some cases, it is useful to select a special induced lightcone such that it agrees with the conventional expectation for a definition of a lightcone, that is, no point causally in the past/future relative to a point $p$ should be present in the future/past lightcone of $p$. This leads us to the following remark:
\begin{remark}
    It is useful to think of $\hat{J}^{\pm}$ as the most restrictive. By the most restrictive, we mean the choice of point $\tilde p$ which minimizes the span of $\hat{J}^{\pm}$, in particular, such that the tip of the cone lies exactly at $p$, or as close as possible to it. In particular, one may replace Definition~\ref{def:ind} by
    \begin{equation}
        \hat{J}^{\pm}(p) = \min_{\tilde{p}} [J^{\pm}(\tilde{p})\cap \mathcal{S}],
    \end{equation}
    where we defined the minimisation so that the point $\tilde{p}=\tilde{p}_\ast$ found by the minimisation satisfies $\hat J^{\pm}(\tilde{p}_\ast)\subseteq \hat J^{\pm}(\tilde{p}^\prime)$ for any $\tilde{p}^\prime$ satisfying Definition~\ref{def:R}.
\end{remark}
We construct the induced lightcones on a screen $\mathcal{S}$ located on the stretched horizon of pure dS$_3$ defined by a fixed $r\leq 1$ coordinate, and induced from an arbitrary point $\tilde c\in \mathcal{I}^-$ or $\tilde r \in \mathcal{I}^+$. We denote their conformal coordinates as
\begin{equation}
    (\sigma_i,\theta_i,\varphi_i)=(k(\pi/2 - \epsilon_{dS}), \theta_i , \varphi_i),
\end{equation}
where $k=+1$ for $\tilde r$ and $k=-1$ for $\tilde{c}$, and $\epsilon_{dS}$ is the UV cutoff of $\mathcal{I}^{\pm}$.
Moving to embedding coordinates,
\begin{equation}
    (X_0^i,X_A^i)=\frac{1}{\cos(\pi/2-\epsilon_{dS})}(k,\sin\theta_i \cos\varphi_i,\sin\theta_i \sin\varphi_i, \cos\theta_i).
    \label{eq:embed-cr}
\end{equation}
The boundary of the lightcone satisfies $-(X_0-X_0^i)^2 + \sum_{A=1}^3 (X_A-X_A^i)^2 =0$, where $\{X_{M}\}_{M=0,1,2,3}$ are the embedding coordinates of a point on the lightcone, and $\{X_{M}^i\}_{M=0,1,2,3}$ are the embedding coordinates of $\tilde{c}$ or $\tilde{r}$. The intersection between this lightcone and the de Sitter hypersurface \eqref{eq:dS_hyperb} gives to leading order in $\epsilon_{dS}$
\begin{equation}
    X_0 = \sum_{A=1}^3 f_A X_A,
\end{equation}
where $f_A = X_A^i/X_0^i$. The boundary of the induced lightcone is given by the intersection of this hypersurface with the stretched horizon:
\begin{equation}
\label{eq:ind_lc}
    \tan\sigma = kr\frac{\cos(\varphi-\varphi_i)}{\sin\theta_i} \pm \cot\theta_i\sqrt{1-r^2\sin^2(\varphi-\varphi_i)}.
\end{equation}
This gives the formula for the induced lightcone on a screen at a constant $r$ in the static patch in the conformal coordinates $(\sigma,\theta=\pi-\arcsin(r\cos\sigma),\varphi)$ as a function of the location of the {tilded} point $(\sigma_i=k(\pi/2-\epsilon_{dS}),\theta_i,\varphi_i)$ at $\mathcal{I}^\pm$.\footnote{We define the domain of the inverse sine function as $\arcsin x \in [-\pi/2,\pi/2]$.}

\subsection{Induced causality vs local causality}
\label{sec:stretched}

We highlighted in Section~\ref{sec:violation} that causality on the screen leads to results in contradiction with the connected wedge theorem. In Section~\ref{sec:ind_LC}, we introduced the notion of the induced lightcone, which we conjecture to be the right object to consider when computing the causal region of the screen operator encoding a localized wavepacket in the static patch. We will further motivate this in the following paper by showing that this prescription leads to a well-defined connected wedge theorem. In this section, we consider the effect of the location of the holographic screen on induced lightcones.

We construct the lightcone of a point $p$ on the holographic screen located at fixed $r$. The result for $r=1$ was given in equation \eqref{eq:embed-sol}. The computation for $r<1$ is analogous, using the parametrization $X_3^2-X_0^2=1-r^2$, $X_1=r\cos\varphi$ and $X_2=r\sin\varphi$ of the fixed $r$ hypersurface in embedding coordinates. Combined with the bulk lightcone equation, one obtains the general solution for a lightcone from a point
\begin{equation}
    p:\, (\sigma_0,\theta_0=\pi-\arcsin(r\cos\sigma_0),\varphi_0)
\end{equation}
(in the conformal coordinates) on a screen $\mathcal{S}$ at a fixed $r$:
\begin{align}
\begin{split}
    \tan\sigma &= \frac{1}{1-r^2}\biggl[-\tan\sigma_0(1-r^2\cos(\varphi-\varphi_0))\biggr.\\
    &+\biggl.\sqrt{r^2(1-\cos(\varphi-\varphi_0))(1-r^2\cos(\varphi-\varphi_0)+1-r^2)(\tan^2\sigma_0+1-r^2)}\biggr].
    \end{split}
\end{align}
To compare this with the induced lightcone, we note the conformal coordinates of the induced point $\tilde{p}$ associated with $p$:
\begin{align}
\begin{split}
    \tilde{p}: \, (\pm(\pi/2-\epsilon_{dS}),\sigma_0+\arccos(r\cos\sigma_0),\varphi_0)).
    \end{split}
\end{align}
Applying equation \eqref{eq:ind_lc} to $\tilde{p}$, we find the induced lightcone of $p$:
\begin{equation}
\begin{split}
    \tan\sigma =&
    kr\frac{\cos(\varphi-\varphi_0)}{\sin(\sigma_0+\arccos(r\cos\sigma_0))} \\
    &\pm \cot(\sigma_0+\arccos(r\cos\sigma_0))\sqrt{1-r^2\sin^2(\varphi-\varphi_0)}.
    \end{split}
\end{equation}
The lightcones from a point $p$ on $\mathcal{S}$ and induced lightcones for the point $p$ are pictured in Figure~\ref{fig:comp} for different values of $r$. 

\begin{figure}
    \centering
      \includegraphics[width=0.9\linewidth]{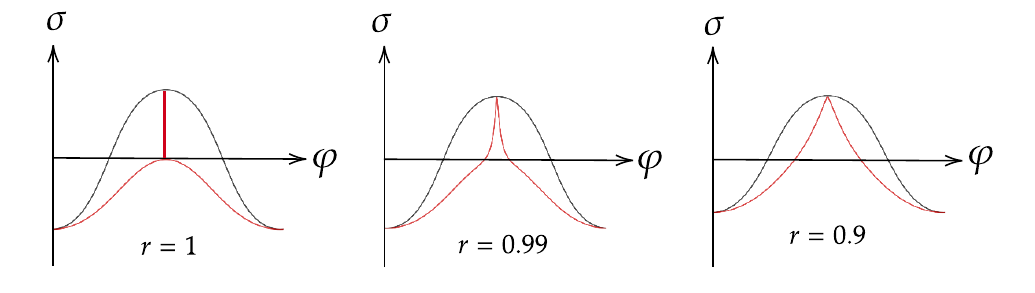}

    \includegraphics[width=0.9\linewidth]{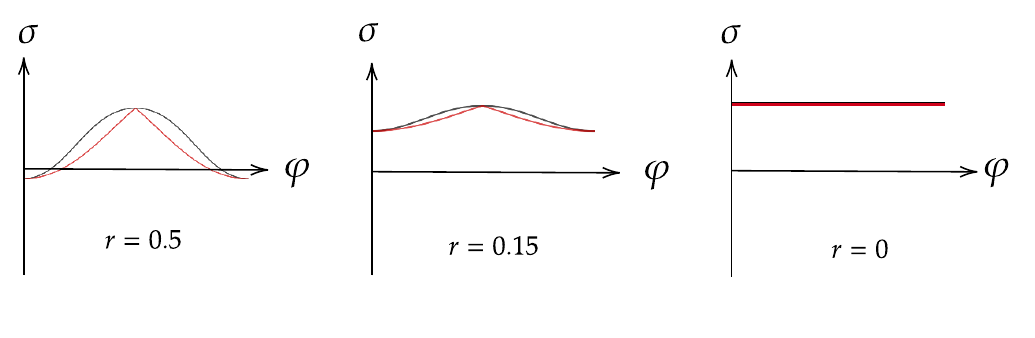}
    \caption{Induced lightcone $\hat{J}^-(p)=\min_{\tilde{p}} J^-(\tilde{p})\cap \mathcal{S}$ (in black) and lightcone from a point $p$ on the screen $J^-\vert_{\mathcal{S}}(p)=J^-(p)\cap\mathcal{S}$ (in red) located at the fixed $r=1,0.99,0.9,0.5,0.15,0$ hypersurface. As $r\rightarrow 0$, the lightcone increases in size, becoming a better approximation of the induced lightcone. Note that at $r=0$ there is no spatial extent in the $\varphi$ direction because it is a pode in dS. }
    \label{fig:comp}
\end{figure}
The induced lightcone $\hat{J}$ always contains the lightcone $J\vert_{\mathcal{S}}$. Moreover, the induced lightcone covers a considerably larger portion of the screen than the lightcone of a point when the screen is located close to the cosmological horizon. As $r\rightarrow 0$, the lightcone of a point spreads and gets closer to the induced lightcone. However, they do not converge, even to the first order in~$r$.

The induced lightcone is interpreted as the causal region associated with the nonlocal operator encoding the local perturbation in the bulk at $p$. The fact that it contains trajectories that are apparently superluminal is due to the nonlocality of this operator, and the fact that it needs some time to be prepared from a local operator. 
On the other hand, the lightcone of a point $p$ is the set of points on the screen that are causally connected to a local operator at $p\in\mathcal{S}$ through the bulk. The difference between the induced lightcone and the lightcone of point $p$ becomes smaller as $r\rightarrow 0$. We interpret this as a localization of the operator on the screen encoding the perturbation at $p$. 

\section{Connected wedge theorem in the static patch}
\label{sec:theorem}
In Section~\ref{sec:violation}, we observed an apparent violation of the connected wedge theorem in the static patch with a specific example. This section is devoted to its resolution by employing induced causality. In particular, we show that the decision regions enlarge with the induced causality prescription. This enables us to recover the connected wedge between them as expected. We then give a bulk proof of the connected wedge theorem in the interior of the holographic screen in asymptotically dS$_3$ spacetime. The proof relies on the static patch holographic conjecture (Conjecture~\ref{conj:SPH}) as well as the focusing theorem and the second law for causal horizons.\footnote{When we also include quantum matter corrections, these theorems are replaced by the restricted quantum focusing conjecture and the generalized second law.}

\subsection{Resolution of the apparent contradiction}
\label{sec:resolution}

We reconsider the example of the 2-to-2 scattering of Section~\ref{sec:violation}. We now apply the prescription of induced lightcones from $\mathcal{I}^{\pm}$ for a screen located at the cosmological horizon $r=1$. The equation of $\partial \hat{J}^{\pm}(p)$ \eqref{eq:ind_lc} simplifies to
\begin{equation}
\label{eq:ind}
    \tan\sigma = \left\{
    \begin{array}{ll}
        k\tan\frac{\theta_i}{2}\cos(\varphi-\varphi_i) & \mbox{if } \sigma k \geq 0 \\
        k\cot\frac{\theta_i}{2}\cos(\varphi-\varphi_i) & \mbox{if } \sigma k \leq 0 
    \end{array}
\right..
\end{equation}
We choose `fictitious' points $\tilde{c}_i$ and $\tilde{r}_i$ by merely extrapolating the lightlike signals emanating from $c_i$ or reaching $r_i$. In this case, these points are 
\begin{align}
\label{eq:fic_coord}
\begin{alignedat}{2}
    \Tilde{c}_1 &=(-\pi/2+\epsilon_{dS},\pi/2+\epsilon_{dS},\pi/2), &\quad \Tilde{c}_2 &=(-\pi/2+\epsilon_{dS},\pi/2+\epsilon_{dS},-\pi/2), \\
    \Tilde{r}_1 &=(\pi/2-\epsilon_{dS},\pi/2+\epsilon_{dS},0), &\quad \Tilde{r}_2 &=(\pi/2-\epsilon_{dS},\pi/2+\epsilon_{dS},\pi).
\end{alignedat}
\end{align}
It is easy to check that this choice of `fictitious' points is the minimal one, as $J^+(\tilde c_1)\cap J^+(\tilde c_1) \cap J^-(\tilde r_1) \cap J^-(\tilde r_2) = (0,\pi/2,0) = J_{12\rightarrow 12}$. We now modify the previous definition of decision regions to accommodate the notion of induced lightcones.
\begin{definition}
\label{def:decision}
    The (induced) decision regions $\mathcal{R}_1$ and $\mathcal{R}_2$ associated with the $2$-to-$2$ scattering $c_1,c_2\rightarrow r_1,r_2 \in \mathcal{S}$ are defined as
    \begin{equation}
    \mathcal{R}_i = \hat{J}^+(c_i)\cap \hat{J}^-(r_1) \cap \hat{J}^-(r_2).
    \end{equation}
Let us denote the largest spatial section of $\mathcal{R}_i$ by $\mathcal{V}_i$, such that $\mathcal{R}_i=\hat{D}(\mathcal{V}_i)$ where we defined $\hat{D}(\mathcal{V}_i)$ as the set of points $p\in\mathcal{S}$ such that $\mathcal{V}_i\in \hat{J}^-(p)\cup\hat{J}^+(p)$. 
\end{definition}
In the example \eqref{eq:Pscat}, we get the following edges of decision region $\mathcal{R}_1$ and $\mathcal{R}_2$, see Figure~\ref{fig:comparison} for an illustration.
\begin{figure}[h]
    \centering
    \includegraphics[width=0.9\linewidth]{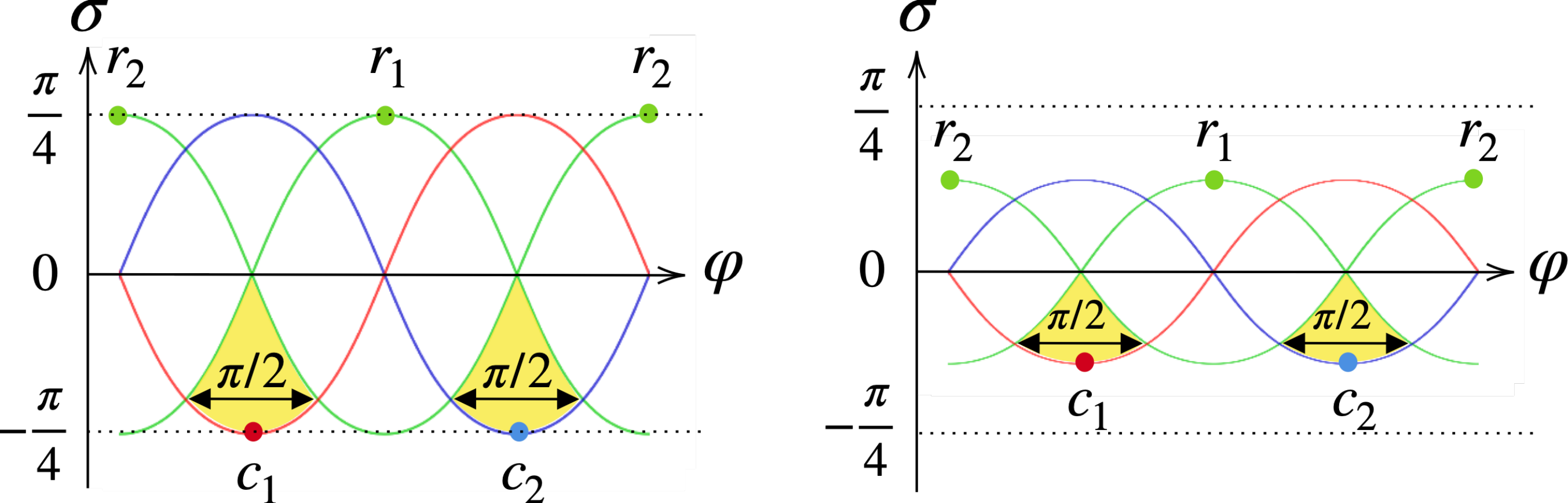}
    \caption{The decision regions $\mathcal{R}_{1,2}$ for the scattering \eqref{eq:Pscat} on a screen located on the cosmological horizon (left diagram) and on a stretched horizon located at $r=0.5$ (right diagram) extend over $\pi/2$ when their lightcones are induced from the input/output points on the null infinities. $\partial \hat{J}^+(\tilde{c}_1)$ is denoted by a red line, $\partial \hat{J}^+(\tilde{c}_2)$ is denoted by a blue line, and $\partial \hat{J}^-(\tilde{r}_{1,2})$ are denoted by green lines.}
    \label{fig:comparison}
\end{figure}
\begin{align}
\begin{split}
\label{eq:V1V2}
  \mathcal{V}_1:\quad & \sigma=-\arctan(1/\sqrt{2})~\quad \varphi\in [\pi/4,3\pi/4],\\
  \mathcal{V}_2: \quad& \sigma=-\arctan(1/\sqrt{2}) ~\quad \varphi\in [-3\pi/4,-\pi/4].
  \end{split}
\end{align}
A direct scattering on $\mathcal{S}$ is not possible from the induced screen causality, as $\mathcal{R}_1\cap\mathcal{R}_2=\varnothing$. Now that $\mathcal{V}_1$ and $\mathcal{V}_2$ are specified, their entanglement entropy can be computed. As reviewed in Appendix~\ref{sec:FPRT}, the entanglement entropy of a spatial subsystem of the screen $\mathcal{S}$ has three contributions. The entropy is given by the sum of the areas of the homologous minimal extremal surfaces 1) in the interior of $\mathcal{S}$, 2) in the interior of a complementary screen associated with an antipodal observer, and 3) in the exterior region bounded by the union of these two screens. These three regions are illustrated in gray in Figure~\ref{fig:penrose}. We denote $D$ the causal diamond inside $\mathcal{S}$ in which we extremize the area.

There are two candidate extremal surfaces of $\mathcal{V}_1\cup\mathcal{V}_2$ in the interior of $\mathcal{S}$. The first one, associated with a disconnected entanglement wedge, is the union of the extremal surfaces for $\mathcal{V}_1$ and $\mathcal{V}_2$, namely, 
\begin{equation}
    \gamma_e^{dis}(\mathcal{V}_1\cup\mathcal{V}_2;D) = \gamma_e(\mathcal{V}_1;D) \cup \gamma_e(\mathcal{V}_2;D),
\end{equation}
where $\gamma_e(\mathcal{V}_i;D)$ is the geodesic anchored to $\mathcal{V}_i$ in the region inside $\mathcal{S}$, whose length is $\min(\Delta\varphi,2\pi-\Delta\varphi)$ where $\Delta\varphi$ is the angle of the arc.\footnote{See appendix B.1 of~\cite{Franken:2023pni} for a detailed computation.} See Figure~\ref{fig:homol-ex} for the extremal surface and the homology region in the interior in each case.
\begin{figure}
    \centering
    \includegraphics[width=0.8\linewidth]{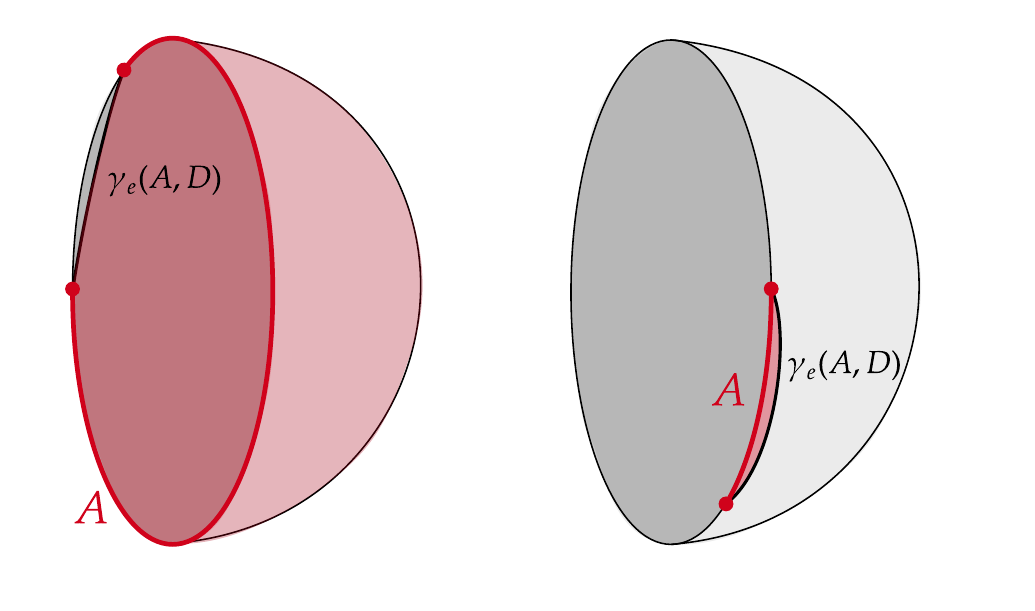}
    \caption{The extremal surfaces and the corresponding homology regions that compute the interior contribution of the holographic entanglement entropy. The time slice of the interior of $\mathcal{S}$ is shown as a semisphere and a subsystem $A\in\Sigma\vert_\mathcal{S}$ is shown as a red arc of length $\Delta\phi$. In the left picture $\Delta\phi \geq \pi$ and in the right picture $\Delta\phi \leq \pi$.The extremal surface $\gamma_e(A;D)$ (thick black curve) is either close to the arc itself or its complement.\protect\footnote{In the limit where $\mathcal{S}$ is located exactly on the cosmological horizon, the extremal surface lies on the horizon \cite{Franken:2023pni}} The homology region $\mathcal{C}_A$ is shaded in red.}
    \label{fig:homol-ex}
\end{figure}
There is a second candidate of extremal surface in the interior of $\mathcal{S}$, given by the extremal surfaces associated with the complement of $\mathcal{V}_1\cup\mathcal{V}_2$ on $\mathcal{S}$. These extremal surfaces are homologous to $\mathcal{V}_1\cup\mathcal{V}_2$, since the union of $\mathcal{V}_1\cup\mathcal{V}_2$ with its complement defines a slice of~$\mathcal{S}$.\footnote{Note that, unlike AdS, the spatial edge of the Penrose diagram of an asymptotically dS is a pole of a sphere and not a boundary.} The associated entanglement wedge is connected through the interior of $\mathcal{S}$, namely,
\begin{equation}
    \gamma_e^{con}(\mathcal{V}_1\cup\mathcal{V}_2;D) = \gamma_e^{dis}((\mathcal{V}_1\cup\mathcal{V}_2)^c;D),
\end{equation}
where $(\cdot)^c$ represents a codimension-two complementary subregion on an achronal surface on the (single) screen containing the subregion.

The exterior contribution comes from the arcs themselves, as the existence of the complementary screen prevents $\mathcal{V}_1\cup\mathcal{V}_2$ from being homologously connected to their complement $(\mathcal{V}_1\cup\mathcal{V}_2)^c$.
In other words, the exterior part is mixed after tracing out the complementary screen while the (right) static patch remains pure.
Essentially, the state dual to $\Sigma$ in the bulk is decomposed as $\ket{\psi}_R\otimes\ket{\psi^c}_{LE}$ up to a local isometry on the right screen. Here, the bulk effective quantum state in $D$ is described by $\ket{\psi}$ in a suitable code subspace, and a quantum state in the causal complement is described by $\ket{\psi^c}$ in another suitable code subspace. One can see that by tracing out the degrees of freedom on the left screen ($L$ and a part of $E$), the state corresponding to the right static patch remains pure as $\dyad{\psi}$ while that corresponding to the exterior becomes mixed by the partial trace.

After all, due to the homology condition, the entanglement wedge of two disjoint subsystems $\mathcal{V}_1\cup\mathcal{V}_2$ on a single screen $\mathcal{S}$ is always disconnected. After the extremization within the exterior domain, the exterior contribution to the entropy (multiplied by $4G_N$) turns out to be just the sum of the lengths of two arcs.

Finally, $\mathcal{V}_1\cup\mathcal{V}_2$ does not extend on the complementary screen so that there is no entropy contribution from the complementary static patch. Thus, the entanglement wedge of $\mathcal{V}_1\cup\mathcal{V}_2$ can only be connected through the interior of $\mathcal{S}$, if $\gamma_e^{con}$ is smaller than $\gamma_e^{dis}$. 

Considering the scattering \eqref{eq:Pscat} and associated induced decision regions \eqref{eq:V1V2}, we find
\begin{equation}
    {\rm area}(\gamma_e^{dis}(\mathcal{V}_1\cup\mathcal{V}_2;D_L)) -{\rm area}(\gamma_e^{con}(\mathcal{V}_1\cup\mathcal{V}_2;D_L))= 0.
\end{equation}
In other words, the scattering we considered corresponds to the transition case where the disconnected and connected entanglement wedges are equivalent. Even though we considered input and output points on a screen located on the cosmological horizon, this result would have been identical for any set of input/output points on an arbitrary stretched horizon at fixed radius $r$ such that $\tilde{c}_i$ and $\tilde{r}_i$ of equation \eqref{eq:fic_coord} are their associated set of {tilded} points on $\mathcal{I}^{\pm}$. Indeed, for these {tilded} points, equation \eqref{eq:ind_lc} becomes linear in $r$, such that the size of $\mathcal{V}_1$ and $\mathcal{V}_2$ is independent of $r$. This is pictured in Figure~\ref{fig:comparison}. Hence, the fact that we found an exact match between the lengths of the connected and disconnected geodesics for these induced points is not specific to a screen on the horizon, but a general property of these `fictitious' points that are associated with a $J_{12\rightarrow 12}$ reducing to a point. This is consistent with the connected wedge theorem, and the typical behavior of pointlike holographic scatterings where the lengths of the connected and disconnected geodesics are equivalent~\cite{May:2019odp, May:2019yxi, May:2021nrl, May:2022clu}. This example thus provides evidence that the induced lightcones defined in Section~\ref{sec:ind_LC} are the right objects to consider when studying causality on the holographic screen.

\subsection{Geometric proof}
\label{sec:proof}

Inspired by the definition of the induced lightcone from points at $\mathcal{I}^{\pm}$, we generalize the result obtained in the last section. The proof of our statement closely follows that of~\cite{May:2019odp, Mori:2023swn}. The statement of the connected wedge theorem in the context of static patch holography goes as follows.

\begin{theorem}[Static patch connected wedge theorem]
 \label{th:CWTpatch}
Let $\mathcal{S}$ be the holographic screen of an observer in an asymptotically dS$_3$ spacetime. Assuming static patch holography (Conjecture~\ref{conj:SPH}), if the $2$-to-$2$ scattering $c_1,c_2 \in \mathcal{S} \rightarrow r_1,r_2\in\mathcal{S}$ is possible in the bulk,
\begin{equation}
    J_{12\rightarrow 12} \neq \varnothing,
\end{equation}
and not on $\mathcal{S}$, 
\begin{equation}
\label{eq:nosc}
    \hat{J}_{12\rightarrow 12} = \hat{J}^+(c_1)\cap \hat{J}^+(c_2)\cap \hat{J}^-(r_1)\cap \hat{J}^-(r_2) =\varnothing,
\end{equation}
then $\mathcal R_1$ and $\mathcal R_2$ have a mutual information $O(1/G_N)$, and their entanglement wedge is connected in the interior of $\mathcal{S}$. We assume that $\mathcal R_1$ and $\mathcal R_2$ consist of connected regions.\footnote{We assume this because some obstacles with a causal horizon in $J^-(\tilde{r}_{1,2})\cap J^+(\tilde{c}_{1,2})$ may cause the decision regions split apart, invalidating our proof below~\cite{Mori:2023swn}.}
\end{theorem}

\begin{proof}
    We will use three important properties here. 
    \begin{itemize}
        \item In the bulk region inside $\mathcal{S}$, minimal extremal surfaces are maximin surfaces and conversely. We develop this statement in Theorem~\ref{th:patch} of Appendix~\ref{app:proofs}.
        \item A congruence of lightrays emanating orthogonally from an extremal surface is of non-positive expansion. This follows from the focusing theorem which states that, under the null energy condition, the expansion parameter $\theta$ of a congruence of lightrays satisfies
        \begin{equation}
            \frac{\d\theta}{\d\lambda}\leq 0,
        \end{equation}
        where $\lambda$ is an affine parameter of the congruence. By definition, $\theta=0$ on an extremal surface. It follows that $\theta\leq 0$ on any null congruence emanating from it. 
        \item The second law of causal horizons~\cite{Jacobson:2003wv}: A causal horizon is the boundary of the causal past of a timelike worldline ending at time infinity $\mathcal{I}^+$. The second law states that the area of such a horizon cannot decrease in time.
    \end{itemize}

Let $\mathcal{V}_i$ be the achronal codimension-two surface on $\mathcal{S}$ such that $\hat{D}(\mathcal{V}_i)=\mathcal{R}_i$, as in Definition~\ref{def:decision}. $\Sigma\vert_\mathcal{S}$ denotes an achronal slice of $\mathcal{S}$ and $\Sigma$ denotes an achronal slice containing the observer of interest in the bulk such that $\partial\Sigma=\Sigma\vert_\mathcal{S}$.
$D=D(\Sigma)$ denotes the causal diamond of slices in the static patch that is bounded by $\Sigma\vert_\mathcal{S}$. Note that the definition of $\Sigma$ here is either $\Sigma_L$ or $\Sigma_R$, which is different from the convention used in Appendix~\ref{app:ext_proc}.
The exact location of the complementary screen $\mathcal{S}_R$ is not important here. 
Additionally, the exterior contribution to the entropy is irrelevant here as the exterior entanglement wedge is always disconnected. The contribution from the complementary patch is absent because the decision regions do not lie on the complementary screen. For a more detailed discussion, see Section~\ref{sec:resolution}.
Let us denote $\gamma_i=\gamma_e(\mathcal{V}_i;D)$ with $i=1,2$ and $\gamma_{1\cup 2}=\gamma_e(\mathcal{V}_1\cup\mathcal{V}_2;D)$ for simplification. To show Theorem~\ref{th:CWTpatch}, we follow the strategy of~\cite{May:2019odp}. We need to prove that 
\begin{equation}
    {\rm area}(\gamma_1\cup\gamma_2) \leq {\rm area}(\gamma_{1\cup 2}).
\end{equation}
If this is true, the entanglement wedge is connected and mutual information of order $O(1/G_N)$. 

For any $\Sigma$, let us construct a codimension-one surface, called the null membrane $\mathscr{N}_{\Sigma}$. See Figure~\ref{fig:membrane} for a sketch.
\begin{figure}
    \centering
    \includegraphics[width=0.5\linewidth]{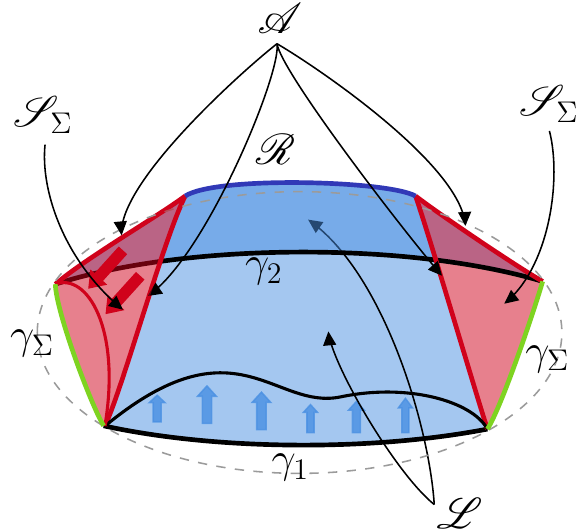}
    \caption{Sketch of the null membrane $\mathscr{N}_{\Sigma}=\mathscr{L}\cup\mathscr{S}_{\Sigma}$ created from $\gamma_1\cup\gamma_2$, depicted by two thick black lines. The slice $\Sigma$ is a disc-like surface bounded by the dashed line representing $\Sigma\vert_{\mathcal{S}}$. The lift $\mathscr{L}$ is depicted by a light blue surface, while the slope $\mathscr{S}_\Sigma$ is depicted by two light red surfaces. The ridge $\mathscr{R}$ is the thick blue line, and the four spacelike cusps $\mathscr{A}=\mathscr{L}\cap\mathscr{S}_{\Sigma}$ are depicted by thick red lines. The surface $\gamma_{\Sigma}$ leading to the contradiction is depicted by two green lines. A deformation of $\gamma_1$ on the lift is depicted by a thin black curve, with the direction of contraction of lightrays depicted by blue arrows. Similarly, a deformation of a set of two connected cusps is depicted by a thin red curve, with the direction of contraction of lightrays depicted by red arrows.}
    \label{fig:membrane}
\end{figure}
$\mathscr{N}_{\Sigma}$ is constructed from the union of null surfaces. The first one, called the lift is
\begin{equation}
    \mathscr{L} = \partial J^+_{in}(\gamma_1\cup\gamma_2) \cap J^-(\tilde{r}_1) \cap J^-(\tilde{r}_1),
\end{equation}
where $J^+_{in}(\gamma_1\cup\gamma_2)$ refers to the congruence of lightrays emanating orthogonally from $\gamma_1\cup\gamma_2$ and directed towards the interior of $D$. The lift $\mathscr{L}$, therefore, consists of the portion of the union of the lightsheets emanating from $\gamma_1$ and $\gamma_2$ that lies in the past of both points $\tilde{r}_1$ and $\tilde{r}_2$, and ends at their meeting points. These meeting points form a spacelike codimension-two surface called the ridge $\mathscr{R}$.

Let us show that the ridge is non-empty. It was shown in~\cite{Wall:2012uf} that whenever an extremal surface is also a maximin surface, the surface must lie outside the causal wedge. This implies that
\begin{equation}
\label{eq:P}
    J_{12\rightarrow 12} \subseteq J^+(\mathcal{C}_1\cup \mathcal{C}_2) \cap J^-(\tilde{r}_1) \cap J^-(\tilde{r}_1),
\end{equation}
where $\mathcal{C}_1,\mathcal{C}_2\in \Sigma$ are the homology regions associated with $\gamma_1$ and $\gamma_2$ on $\Sigma$.\footnote{The homology region is a spacelike slice bounded by the extremal surface and its associated subsystem of the screen. Its causal diamond gives the corresponding entanglement wedge. See Definition~\ref{def:homol}.} $J_{12\rightarrow 12}$ is assumed to be non-empty, such that the right-hand side of equation \eqref{eq:P} is also non-empty. This region being not empty implies that a subset of $J^+(\mathcal{C}_1\cup \mathcal{C}_2)$ exists in the past of $\tilde{r}_1$ and $\tilde{r}_2$. Because $J^+(\mathcal{C}_1\cup \mathcal{C}_2) = J^+(\mathcal{C}_1) \cap J^+(\mathcal{C}_2)$, $J^+(\mathcal{C}_1)$ and $J^+(\mathcal{C}_2)$ must meet in the past of  $\tilde{r}_1$ and $\tilde{r}_2$. The ridge is therefore non-empty.

The other constituent of $\mathscr{N}_{\Sigma}$, called the slope $\mathscr{S}_{\Sigma}$, is defined as
\begin{equation}
    \mathscr{S}_{\Sigma} = \partial[J^-(\tilde{r}_1)\cup J^-(\tilde{r}_2)] \cap J^-[\partial J^+_{in}(\gamma_1\cup\gamma_2)] \cap J^+(\Sigma).
\end{equation}
The slope is a subsystem of two causal horizons associated with $\tilde{r}_1$ and $\tilde{r}_2$ on $\mathcal{I}^+$, which is in the future of $\Sigma$ and the past of the lightsheets emanating from $\gamma_1$ and $\gamma_2$. We now construct the null membrane as
\begin{equation}
    \mathscr{N}_{\Sigma} = \mathscr{S}_{\Sigma} \cup \mathscr{L}.
\end{equation}
Let $\gamma_{\Sigma}$ be the past boundary of the slope, that is
\begin{equation}
     \gamma_{\Sigma} = \partial[J^-(\tilde{r}_1)\cup J^-(\tilde{r}_2)] \cap J^-[\partial J^+_{in}(\gamma_1\cup\gamma_2)] \cap \Sigma.
\end{equation}
The boundary of $\gamma_{\Sigma}$ must be the same as the one of $\gamma_1\cup\gamma_2$, as the latter is defined to be on the past lightcone of $\tilde{r}_1$ and $\tilde{r}_2$.\footnote{Note that the boundary of $\gamma_1\cup\gamma_2$ is the same as the boundary of $\gamma_{1\cup 2}$ by the homology condition.} In particular, the intersection of these past lightcones with $\partial\Sigma=\Sigma\vert_\mathcal{S}$ must therefore coincide with $\partial \mathcal{V}_i = \partial \gamma_i$. Hence, $\gamma_{\Sigma}\cup\gamma_1\cup\gamma_2$ is a closed codimension-two surface on $\Sigma$.\footnote{The connectivity of $\gamma_{\Sigma}\cup\gamma_1\cup\gamma_2$ is ensured by hyperbolicity of asymptotically de Sitter spacetime.}

The slope and the lift intersect on a set of four spacelike cusps, denoted by $\mathscr{A}$. It must exist, otherwise, it would imply that the past lightcones of $\tilde{r}_1$ and $\tilde{r}_2$ have intersected and hit the boundary before reaching the ridge $\mathscr{R}$, meaning $\mathcal R_1$ and $\mathcal R_2$ have an intersection, such that the theorem is trivially satisfied. The lift is a subsystem of a lightsheet (since it emanated from an extremal surface), it is of negative expansion. Therefore,
\begin{equation}
\label{eq:lift}
    \text{Area}(\gamma_1\cup\gamma_2) \geq \text{Area}(\mathscr{A}) + 2\text{Area}(\mathscr{R}).
\end{equation}
The second law of causal horizons implies that $\text{Area}(\mathscr{A})\geq \text{Area}(\gamma_{\Sigma})$. Combining this with equation \eqref{eq:lift}, we get
\begin{equation}
\label{eq:slope}
    \text{Area}(\gamma_1\cup\gamma_2) \geq \text{Area}(\gamma_{\Sigma}).
\end{equation}
We have constructed on every slice $\Sigma$ a codimension-two surface homologous to $\mathcal{V}_1\cup\mathcal{V}_2$ with a smaller area than the extremal area surface $\gamma_1\cup\gamma_2$ with a disconnected entanglement wedge, implying that $\gamma_1\cup\gamma_2$ is not a true maximin surface,
since it is not minimal on any slice $\Sigma$. Minimal extremal surfaces in the interior of $D$ are maximin surfaces (Theorem~\ref{th:patch}) so $\gamma_1\cup\gamma_2$ cannot be the smallest extremal surface, concluding the proof.
\end{proof}

\begin{remark}
    Theorem~\ref{th:CWTpatch} and its proof may be generalized to semiclassical spacetimes, where holographic entanglement entropy includes quantum matter corrections, by replacing everywhere the area of surfaces by their generalized entropy~\eqref{eq:Sgen}, and assuming the restricted quantum focusing conjecture~\cite{Shahbazi-Moghaddam:2022hbw, Bousso:2015mna} as well as the generalized second law of causal horizons~\cite{Wall:2009wm,Wall:2011hj}.
\end{remark}

\section{Discussion}
\label{sec:discussion}
In the present paper, we explore causality on the screen in asymptotically de Sitter spacetime through holographic scattering within the framework of static patch holography. The main takeaway of the paper is that the causality on the screen consistent with the connected wedge theorem is induced by local excitations at null infinities. Using this induced causality, we can provide a holographic proof of the connected wedge theorem in the static patch of an observer. Let us conclude this paper with some remarks on dS holography and future directions.

\paragraph{Precise nature of excitations on a screen}
In this paper, we did not discuss the detailed nature of a local excitation at null infinity or `smeared' excitation on the screen. In fact, knowing what types of excitations are allowed is very important to support our proof of the connected wedge theorem~\cite{Mori:2023swn}. Nevertheless, at this point the precise holographic dual of de Sitter, in particular within static patch holography, is unclear. The DSSYK model~\cite{Susskind:2021esx,Lin:2022nss,Rahman:2022jsf,Goel:2023svz,Narovlansky:2023lfz,Verlinde:2024znh,Verlinde:2024zrh,Blommaert:2023opb, Blommaert:2023wad,Rahman:2024iiu} is one candidate, however, the precise location of the screen is still debatable. Additionally, it is dimensionally reduced so the dual of dS$_3$ is further unclear. In the proposed dS/DSSYK duality, Susskind has argued that only a small number of special collective excitations of the fermions in the DSSYK model would propagate deep in the bulk \cite{Susskind:2022bia}. It is interesting to see what type of excitations in the theory realize holographic scattering if there are, and see if the model shows signs of induced causality. We can already point out that the induced lightcones constructed in this paper showcase a difference between perturbations propagating into the bulk and perturbations that are confined close to the stretched horizon. Indeed, the induced lightcone of a bulk propagating perturbation is apparently superluminal while it does not violate the microcausality~\cite{Mori:2023swn}. This means such an excitation is nonlocal and obeys a nonlocal evolution from the screen point of view. On the other hand, the lightcone of a perturbation confined on the screen is a lightcone based on the induced metric, and it does not exhibit superluminality. A further investigation of these perturbations and their duals in a concrete theory could help us to identify the true UV boundary, from which holographic dS emerges.

Another possible direction is to construct a finite-dimensional toy model. Since quantum tasks in nonlocal quantum computations are better understood in finite dimensions, often by using circuit diagrams, we may be able to discuss a general feature of holographic quantum tasks by modeling static patch holography by a (random) quantum error correcting code as in AdS/CFT~. This also naturally provides a holographic dictionary {between} the effective picture, where a local unitary via {a} direct scattering takes place, and the fundamental picture, where the nonlocal quantum computation {may take} place~\cite{Akers:2022qdl,Akers:2021fut}.

\paragraph{Pushing the screen inside the static patch}

We briefly commented in Section~\ref{sec:stretched} on the effect of the location of the screen on causality. 
We found that pushing the screen deeper into the static patch tends to localize the effective holographic theory, as the causal region tied to a bulk perturbation approaches that of a local operator near the observer’s worldline.
In general, we expect that the effective holographic theory defined on the screen follows some kind of renormalization group flow as one pushes the screen closer to the observer's worldline. It would be interesting to describe the details of the coarse-graining associated with moving the holographic screen in the static patch. In parallel, one might consider holographic scattering in alternative static patch holographic setups, such as worldline holography or half-de Sitter holography~\cite{Anninos:2011af,Banihashemi:2022htw,Kawamoto:2023nki}.

\paragraph{Scattering between holographic screens}

In this work, we only considered holographic scatterings in the static patch. Since global de Sitter spacetime is conjectured to be encoded on the holographic screens associated with two antipodal observers, an interesting direction would be to consider scatterings connecting the two screens to probe the exterior region. One clear obstruction is that no direct scattering is possible in the static patch, due to cosmological horizons between the screens. One way to overcome this issue would be to consider explicit solutions of asymptotically de Sitter spacetimes in which the static patches always overlap~\cite{Gao:2000ga,Leblond:2002ns}. Another possibility would be to use the mapping between perturbations on the screen and operators at null infinity. One can define scattering from two points on two disconnected holographic screens to two points at future null infinity. We hope that it is possible to find a precise mapping between localized output points at $\mathcal{I}^+$ and extended regions on the holographic screens by making use of entanglement between the screens, although the details of this mapping are yet to be explored. In addition to providing additional evidence for the connection between static patch holography and dS/CFT, scattering in the exterior region may provide a useful tool in determining the correct entanglement entropy prescription in the region between the screen.

An implicit assumption made in the connected wedge theorem is that the evolution on the holographic screen(s) is local. Otherwise, the time evolution manifestly entangles two decision regions so it cannot be viewed as a nonlocal quantum computation.\footnote{We thank Beni Yoshida for pointing this out.} The Hamiltonian generating time evolution in the two static patches is $H=H_R-H_L$, where $H_{L,R}$ are the Hamiltonian for each static patch. On the other hand, the screen dual of the evolution in the Milne patch toward $\mathcal{I}^+$ {may require coupling the left and the right boundaries similar to the eternal AdS black hole.} Thus, it is essential to understand the locality of the time evolution on the screens to justify the connected wedge theorem in the exterior region.

\paragraph{Relating static patch holography to dS/CFT}

In the de Sitter version of the connected wedge theorem argued in this paper, we related input and output points to points at past and future null infinities. This maps a bulk scattering from the static patch holographic screens to a bulk scattering from conformal boundaries. It would be interesting if the dS/CFT correspondence could provide a precise framework to describe this type of scattering. In particular, can we prove a connected wedge theorem in dS/CFT? Doing so would require precise notions of causality on the Euclidean CFT, which is problematic as there is no notion of time evolution. 
Our construction may serve as an alternative definition of `time' in dS/CFT as it provides a mapping between points on the static patch, which follows the real time evolution, to the points on the conformal boundary. See Figure~\ref{fig:mapping} for its picture. 
Additionally, studying connected wedges in dS/CFT would require a precise prescription to compute the entanglement entropy of CFT subregions. This has been studied in~\cite{Doi:2022iyj,Narayan:2015vda,Narayan:2020nsc,Narayan:2022afv, Narayan:2023zen,Doi:2023zaf, Das:2023yyl,Doi:2024nty}. 
Succeeding in this challenge would open a path to a better understanding of time, entanglement, and causality in dS/CFT.

Another interesting aspect hinted at in this paper is the bulk reconstruction. By extending the holographic scattering on the screen to null infinities, the corresponding four-point correlators should share the same bulk point singularity~\cite{Gary:2009ae,Heemskerk:2009pn,Penedones:2010ue,Maldacena:2015iua}. This operator perspective suggests that there might be a mapping between a nonlocal operator or an operator prepared by a finite time on the screen in static patch holography and local operator(s) on some Euclidean CFT on null infinity in dS/CFT. Is there any relation between two dS holography proposals at the operator level? This question reminds us of the HKLL reconstruction in dS/CFT, where a bulk local operator at point $p$ is reconstructed from CFT operators supported on $\mathcal{I}^+\cap J^+(p)$~\cite{Goldar:2024crc}. 
In the present case, we consider a fixed-momentum operator so the support of the smearing function on null infinity seems to be consistent with our expectation that the `fictitious' excitation is obtained by extending the null ray beyond the screen.
The present work and the notion of induced causality provide a potential guiding line to explore the possible connection between these two previously disconnected approaches to de Sitter holography. 

\emph{Note added} -- The mapping of `time' between an observer worldline in static patch and null infinity is argued in~\cite{Parikh:2002py}.\footnote{We thank Zixia Wei for introducing this reference and a relevant discussion.} The authors focus on flat slicing by starting from dS/$\mathbb{Z}_2$, where the quotient is given by the antipodal identification. The `time' in null infinity is obtained by analytically continuing the timelike Killing vector to the Milne patch. As also mentioned by Strominger \cite{Strominger:2001pn}, this flow corresponds to the scale transformation on null infinity. This gives an interpretation of the state evolution in light of the radial quantization. The setup in~\cite{Parikh:2002py} is different from ours, as dS is a quotient or global and the mapping is to the observer's worldline or a holographic screen. When we focus on the ingoing or outgoing mode of scattering, the mapping only requires flat slicing, thus for each mode, there is no essential difference. However, to argue the decision regions relevant to the connected wedge, we need to combine the future and past induced light cones. This requires us to consider global dS, not a quotient one. Nevertheless, the mapping itself is the same and our discussion on the connected wedge theorem gives another evidence for the mapping proposal.

\paragraph{Generalization to flat space holography}

\begin{figure}
    \centering
    \includegraphics[width=0.7\linewidth]{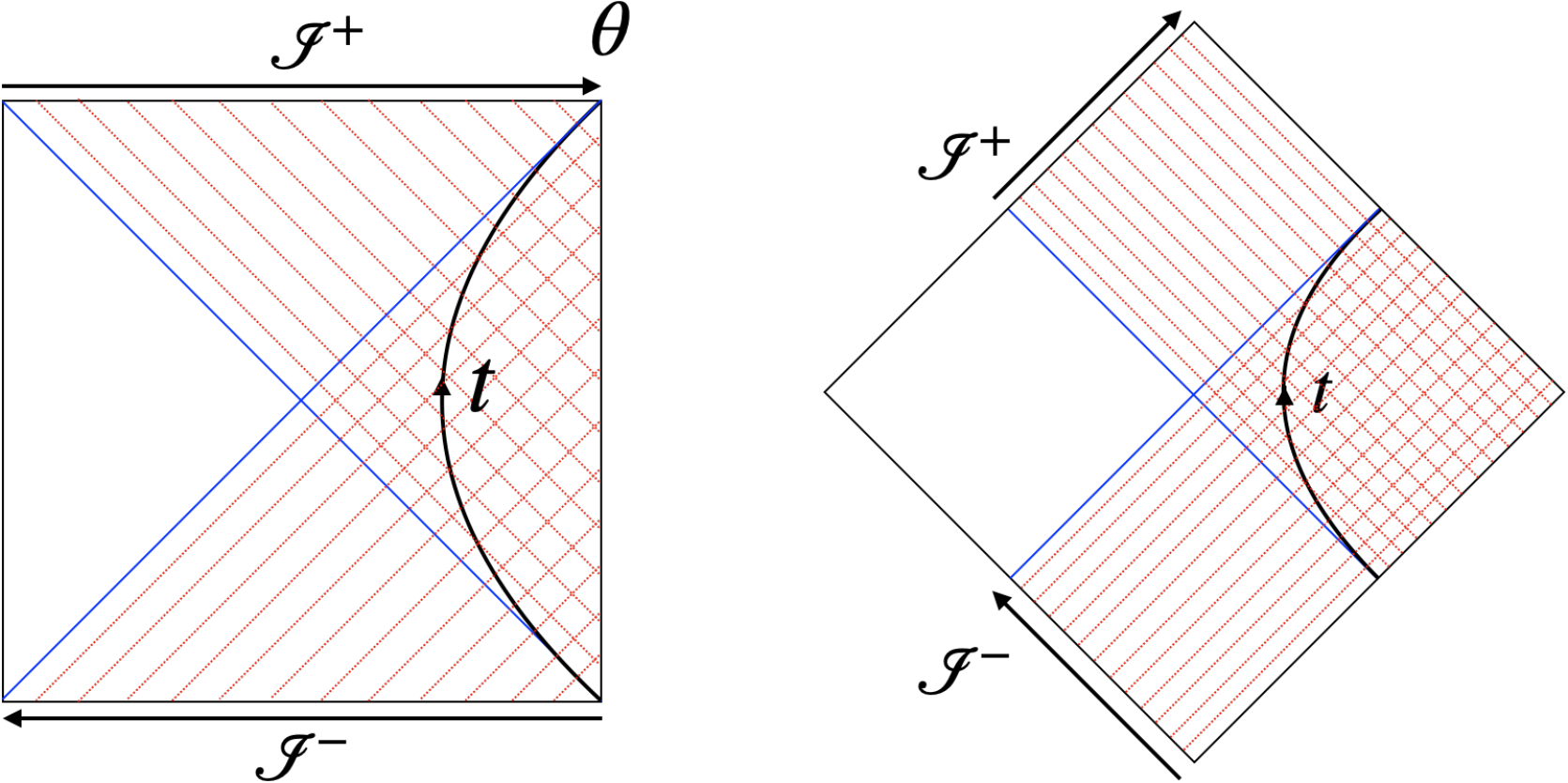}
    \caption{The mapping between points on a holographic screen $\mathcal{S}$ and points on the null infinities $\mathcal{I}^\pm$. Left: The proposed mapping is based on the induced causality in the de Sitter Penrose diagram. The time translation in the static patch is mapped to a spatial translation on $\mathcal{I^\pm}$. Right: A possible generalization to the Minkowski spacetime. The time translation in the Rindler patch is mapped to a null translation (supertranslation) on $\mathcal{I^\pm}$.}
    \label{fig:mapping}
\end{figure}

In this paper, we start from a holographic screen at the cosmological/stretched horizon dual to the interior in the static patch of asymptotically de Sitter spacetimes. Holographic scattering supports causality induced from the null infinities. Considering the most restricted, future, and past induced lightcones, we find a one-to-one correspondence between the points on the timelike screen, extending from $\mathcal{I}^-$ to $\mathcal{I}^+$, to the points on the null infinities $\mathcal{I}^\pm$. This offers a natural way of mapping the Lorentzian time evolution in a static patch to a spatial evolution on the null infinities.

We might ask whether it's possible to extend this to more general spacetimes, including the Minkowski spacetime.\footnote{We thank Sabrina Pasterski for a discussion on this subject.} One straightforward generalization is to consider a Rindler observer in the Minkowski spacetime. Then, the cosmological horizon is identified with the Rindler horizon. By assuming a holographic screen near the Rindler horizon, we can ask if a similar mapping is possible. Indeed, the null extension seems to suggest there could be a dual description on a part of the null boundary. See Figure~\ref{fig:mapping} for the comparison between the de Sitter case (left) and the Minkowski case (right). Like the dS case, the Lorentzian time evolution in the Rindler patch is mapped to an evolution on the null infinities, however, this time the evolution is generated by null generators like supertranslations. If we take the Minkowski spacetime to be four-dimensional, the picture resembles to Bondi-van der Burg-Metzner-Sachs (BMS) field theory, or equivalently, the Carrollian field theory (see~\cite{Donnay:2022aba,Donnay:2022wvx} and references therein).
Furthermore, flat space holography for a wedge region reminds us of the wedge holography from the de Sitter slicing~\cite{Ogawa:2022fhy}, where the wedge region is argued to be dual to a dS braneworld on the screen and it is further reduced to a codimension-two sphere on the null infinities. In our case, the mapping approach suggests a different UV realization on the codimension-one null infinities. Since the mapping stems from scattering in the bulk, this procedure may be compatible with celestial holography~\cite{Pasterski:2016qvg,Pasterski:2017kqt}, which is a conventional approach to flat space holography.

If we focus on the Rindler patch, rather than the whole Minkowski spacetime, the procedure breaks some isometry, so establishing a concrete flat space holography may be more involved. Nevertheless, as we could examine causality and entanglement entropy of a holographic screen from the connected wedge theorem without knowing the details of the theory, we might also be able to probe causality and holographic entanglement entropy prescription in flat space, which has been debated~\cite{Jiang:2017ecm,Apolo:2020bld}, in a similar manner.

\acknowledgments

We are grateful to Sergio Aguilar-Gutierrez, Nirma
lya Kajuri, Cynthia Keeler, Josh Kirklin, Zhi Li, Alex May, Hervé Partouche, Sabrina Pasterski, Jan Pieter van der Schaar, Zixia Wei, and Beni Yoshida for helpful discussions. V.F. and T.M. would like to thank the Berkeley Center for Theoretical Physics for their hospitality during the early stages of this work. V.F. would like to thank Perimeter Institute for their hospitality during the final stage of this work. 

This research is partially supported by the Cyprus Research and Innovation Foundation grant EXCELLENCE/0421/0362. 
This research was also supported in part by the Perimeter Institute for Theoretical Physics. Research at Perimeter Institute is supported by the Government of Canada through the Department of Innovation, Science and Economic Development and by the Province of Ontario through the Ministry of Research, Innovation and Science. This work was supported by JSPS KAKENHI Grant Number 23KJ1154, 24K17047.

\appendix

\section{Summary of notations}
\label{app:notations}

This paper involves several notations.
\begin{itemize}
    \item Plain capital letters $B;D,J$,... denote bulk subregions.
    \item Italic capital letters $\mathcal{C,I,L,R,S}$,... denote codimension-one hypersurfaces.
    \item The lower-case Greek letter $\gamma$ denotes an achronal codimension-two surface.
    \item Lower case letters $c,r,p$,... denote points in the bulk.
    \item Calligraphic capital letters $\mathscr{A,L,N,R,S}$,... denote subsystems of the null membrane, an object used in the proof of Theorem~\ref{th:CWTpatch}.
\end{itemize}
To avoid confusion, we sometimes follow the common notations in the literature. For example, an achronal codimension-one hypersurface is denoted by $\Sigma$; an entanglement wedge, which is codimension-zero, is denoted by $\mathcal{W}$. Some additional notations used in this paper are listed below.
\begin{itemize}
    \item $X_M$ with $M=0,1,2,3$ denotes the embedding coordinates of dS$_3$.
    \item $(\sigma,\theta,\varphi)$ and $(t,r,\varphi)$ are the conformal and static coordinates of dS$_3$, respectively.
    \item  $\mathcal{L}(\gamma)$ denotes the null hypersurface defined by the congruence of geodesics\footnote{A congruence of geodesics is a set of non-intersecting geodesics whose union constitutes an open subregion of a spacetime. Alternatively, it is a set of geodesics passing through an open subregion such that for each point in the subregion, only one geodesic passes through it.} emanating from $\gamma$. $\lambda$ denotes the affine parameter along it, and $\theta$ is its expansion.
    \begin{equation}
        \theta = \frac{1}{\mathcal{A}}\frac{d\mathcal{A}}{d\lambda},
    \end{equation}
    where $\mathcal{A}$ is the infinitesimal area element spanned by nearby null geodesics. A null hypersurface with $\theta\leq 0$ everywhere is called a lightsheet.
    \item $J^{\pm}(p)$ are the past $(-)$ and future $(+)$ lightcones of $p$ in the bulk, and $D(\Sigma)$ is the bulk causal diamond of $\Sigma$.
    \item $\mathcal{I}^{\pm}$ are the past $(-)$ and future $(+)$ null infinities of $dS$. Points or regions defined on $\mathcal{I}^{\pm}$ are {tilded} with $~\tilde~$.
    \item The conditioning $(\cdot)\vert_\mathcal{S}$ will be used from Section~\ref{sec:asympt-task} to describe bulk causality restricted to the boundary/screen $\mathcal{S}$ (except for $\Sigma\vert_{\mathcal{S}}$ introduced later). For a precise definition, see~\eqref{eq:deflc}.
    \item The subscript ${}_{\mathcal{S}}$ (without conditioning as above) will be used from Section~\ref{sec:violation} to describe causality based on the induced metric on the screen. As mentioned in Section~\ref{sec:violation}, this is distinguished from the causality above, denoted by the subscript $(\cdot)\vert_\mathcal{S}$.
    \item The accent $\hat{\phantom{a}}$ will be used in Theorem~\ref{thm:refined-CWT} and from Section~\ref{sec:ind_LC} to describe the notion of causality induced from $\mathcal{I}^{\pm}$. It is also used for some quantity defined from causality, such as the induced causal diamond $\hat{D}$.
    \item $\mathcal{R}_i$ denotes a decision region, a subregion on the screen $\mathcal{S}$, which appears in the connected wedge theorem.
    \item $\mathcal{V}_{i}$ is the largest spatial section of $\mathcal{R}_i$, meaning $\mathcal{R}_i=\hat{D}(\mathcal{V}_i)$.
    \item A spatial subregion on a screen is denoted by $A$ and $\Sigma\vert_{\mathcal{S}}$ denotes the intersection between an achronal slice $\Sigma$ and the screen $\mathcal{S}$. This defines the complementary subregion $A^c$ as $A\cup A^c=\Sigma\vert_\mathcal{S}$.
\end{itemize}

\section{Formal aspects of holographic entropy in de Sitter }
\label{app:ext_proc}

We review the covariant holographic entanglement prescription of~\cite{Franken:2023pni} in Section~\ref{sec:FPRT}. In Appendix~\ref{app:def}, we define the definitions of three alternative prescriptions to find the codimension-two surfaces computing entropies. The first one, extremality, corresponds to the direct adaptation of the Hubeny-Rangamani-Takayanagi (HRT) prescription in AdS/CFT~\cite{Hubeny:2007xt}. The second one, maximin, is the direct adaptation of the maximin procedure which is equivalent to the HRT prescription in AdS spacetime. The third one, C-extremality, is the proposed adaptation of extremality motivated by the non-existence of extremal surfaces in some cases in de Sitter spacetime~\cite{Franken:2023pni}. We prove their inequivalence and that a maximin and C-extremal surface is an extremal surface. We finally show that the three prescriptions are equivalent in the interior of the holographic screen of an observer in asymptotically de Sitter spacetime.

\subsection{Holographic entanglement entropy}
\label{sec:FPRT}

Inspired by the Ryu-Takayanagi (RT) formula and its generalization called the HRT formula~\cite{Ryu:2006bv, Hubeny:2007xt, Wall:2012uf, Faulkner:2013ana, Engelhardt:2014gca}, one can consider the problem of computing entanglement entropies in the holographic dual theory in terms of bulk quantities. This question was tackled in~\cite{Susskind:2021esx, Shaghoulian:2021cef,Shaghoulian:2022fop,Franken:2023pni} and we review the covariant proposal of~\cite{Franken:2023pni}, which is used in Section~\ref{sec:theorem}. 

Considering an arbitrary observer worldline, we can perform a change of coordinates using the isometries of de Sitter spacetime such that this observer lies at the north pole $\theta=0$. One can associate an antipodal observer to this trajectory, located at $\theta=\pi$. Their static patches are complementary in the sense that their union contains an infinite number of complete Cauchy slices.\footnote{These slices all pass through the bifurcate horizon which is the closed surface located at $\sigma=0,\theta=\pi/2$.} The Penrose diagram of pure de Sitter spacetime is depicted in Figure~\ref{fig:penrose}, with related notations.
\begin{figure}
    \centering
    \includegraphics[width=0.4\linewidth]{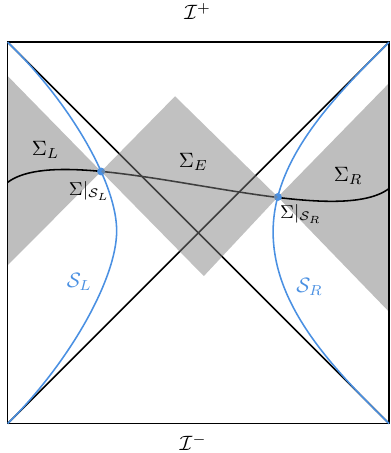}
    \caption{Penrose diagram of dS. The two blue lines are the two screens, and the three parts of a complete slice $\Sigma$ are depicted by black lines. The worldlines of two antipodal observers follow the two vertical lines. The bulk regions $D_I$ are depicted by shaded regions. For $D_E$ to be encoded on $\mathcal{S}_E$, the two screens should be located at a Planck length away from the horizons.}
    \label{fig:penrose}
\end{figure}
The holographic screens associated with two antipodal observers are denoted as $\mathcal{S}_L$ and $\mathcal{S}_R$. Each screen is an arbitrary convex timelike surface that resides in each static patch. In this appendix, we consider global dS, including the region beyond a single static patch. Thus, we slightly generalize the notation such that $\mathcal{S}$ denotes the union of the two screens, i.e. $\mathcal{S}=\mathcal{S}_L \cup \mathcal{S}_R$. We note an arbitrary spacelike slice of $\mathcal{S}_L\cup \mathcal{S}_R$ by $\Sigma\vert_{\mathcal{S}_L}\cup \Sigma\vert_{\mathcal{S}_R}$, associated with a complete achronal slice $\Sigma$. In particular, $\Sigma\vert_{\mathcal{S}_L}$ and $\Sigma\vert_{\mathcal{S}_R}$ are spacelike regions and they are spacelike separated. We denote the boundary of the middle region (exterior) by $\mathcal{S}_E=\mathcal{S}_L\cup\mathcal{S}_R$. We denote $\Sigma_I$, with $I=L,R,E$, any achronal slice such that $\partial \Sigma_I = \Sigma\vert_{\mathcal{S}_I}$ for $i\in\{L,E,R\}$. (Note that $\Sigma$ in the main body is now denoted by $\Sigma_R$.) We denote $D_I=D(\Sigma_I)$ the bulk region spanned by slices $\Sigma_I$. Now let us consider any spacelike subregion $A$ of $\Sigma\vert_{\mathcal{S}_L}\cup\Sigma\vert_{\mathcal{S}_R}$.
\begin{definition}[Homology condition]
\label{def:homol}
    A surface $\gamma\in D$ is $D$-homologous to the codimension-two achronal surface $A \in\partial D$ if there exists an achronal codimension-one surface $\mathcal{C}(A)$ such that $\partial \mathcal{C}(A) =\gamma \cup A$.\footnote{This implies in particular that $\gamma$ is anchored on $ A$, i.e. $\partial \gamma=\partial A$.\label{nb}} Surface $\mathcal{C}(A)$ is called the homology region in $D$. In the case of de Sitter space, we use the abuse of language that $\gamma$ is $D_I$-homologous to $A\in\Sigma\vert_{\mathcal{S}_L}\cup \Sigma\vert_{\mathcal{S}_L}$ if it is $D_I$ homologous to $A_I = A\cap\mathcal{S}_I$. The homology region is then denoted by $\mathcal{C}_I(A)$ and we impose that $\mathcal{C}_I(A)\in\Sigma'_I$ where $\Sigma_I'$ is an achronal slice such that $\partial\Sigma_I' = \Sigma\vert_{\mathcal{S}_I}$.
\end{definition}
\begin{figure}
    \centering
    \includegraphics[width=0.4\linewidth]{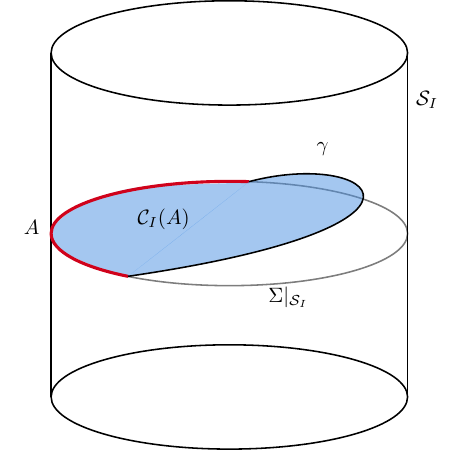}
    \caption{Homology condition in $D_I$ for $I=L$ or $I=R$. The screen $\mathcal{S}_I$ has the topology of a cylinder. $\Sigma\vert_{\mathcal{S}_I}$ has the topology of a circle on the cylinder. $\gamma$ is homologous to $A$ as their union defines the boundary of a spacelike region $\mathcal{C}_I(A)\in\Sigma_I'$ shaded in blue.}
    \label{fig:homology}
\end{figure}
We now review the covariant holographic entanglement entropy proposal from~\cite{Franken:2023pni}, based on~\cite{Shaghoulian:2021cef}. One of the main motivations for this conjecture comes from wedge holography~\cite{Shaghoulian:2021cef, Geng:2020fxl, Mollabashi:2014qfa, Miao:2020oey,Akal:2020wfl, Bousso:2020kmy} and the consistency of its results with the Gibbons-Hawking formula as well as strong subadditivity and entanglement wedge nesting~\cite{Shaghoulian:2021cef,Shaghoulian:2022fop,Franken:2023pni}.
\begin{conjecture}[Covariant holographic entanglement entropy~\cite{Shaghoulian:2021cef,Shaghoulian:2022fop,Franken:2023pni}]
\label{con:ent_pr}
    Let $A$ be a subsystem of $\Sigma\vert_{\mathcal{S}_L}\cup\Sigma\vert_{\mathcal{S}_R}$. The entanglement entropy of $A$ is given by
    \begin{equation}
    	\label{eq:entropy}
    S(A)= \sum_{I=L,E,R}{\frac{{\rm area}(\gamma_e(A;D_I))}{4G_N}} +O(1),
    \end{equation}
    where $\gamma_e(A;D_I)$ is the minimal extremal surface $D_I$-homologous to $A$.\footnote{When there is no such surface in some region $D_I$, one has to extend the notion of extremal surface, see Appendix~\ref{app:def} and~\cite{Franken:2023pni}. We show in Appendix~\ref{app:proofs} that the usual notion of extremality is enough in all cases considered in this paper.} The entanglement wedge $\mathcal{W}(A)$ of $A$ is given by the union of the causal diamonds of all three homology regions $\mathcal{C}_I(A)$. Each of them is denoted as $\mathcal{W}_I(A)=D(\mathcal{C}_I(A))$
    \begin{equation}
        \mathcal{W}(A) = \bigcup_{I} D(\mathcal{C}_I) = \bigcup_I \mathcal{W}_I(A).
    \end{equation}
\end{conjecture}

\begin{remark}
    Quantum corrections due to the entropy of bulk fields can be taken into account by replacing the sum over the area terms with the generalized entropy~\cite{Faulkner:2013ana, Engelhardt:2014gca, Franken:2023pni}
    \begin{equation}
    \label{eq:Sgen}
    S_{gen}(A)= \sum_{I=L,E,R}{\frac{{\rm area}(\gamma_i(A;D_I))}{4G_N}}\\ + S_{\rm bulk}\left(\bigcup_{I}{\cal {C}}_{I}(A) \right).
\end{equation}
The surface extremizing the generalized entropy is called the quantum extremal surface.
\end{remark}
In the definitions above, we considered a restricted version of the usual HRT formula, in analogy with the restricted maximin procedure~\cite{Marolf:2019bgj, Grado-White:2020wlb}, which is analogous to the original maximin definition, except that it considers bulk surfaces defined on achronal slices whose boundary is identified with the screen $\Sigma\vert_{\mathcal{S}}$, instead of slices with the looser condition that the boundary contains the subsystem $A$. In AdS the results agree when the maximin surface lies in a smooth region of spacetime~\cite{Marolf:2019bgj}. The restricted maximin or HRT is necessary in some backgrounds to ensure strong subadditivity and entanglement wedge nesting~\cite{Grado-White:2020wlb}. From another perspective, restricted extremization is important when dealing with systems that are not isolated, such as evaporating black holes where time evolution on the boundary alters the entropy~\cite{Penington:2019npb, Almheiri:2019psf} . Similarly, in our scenario, each region $D_I$ is not isolated, motivating the use of the restricted prescription (Conjecture~\ref{con:ent_pr}). For instance, when considering one holographic screen, applying a unitary transformation to the complementary screen could affect the entropy through the exterior region.

Conjecture~\ref{con:ent_pr} is central in Section~\ref{sec:theorem} when discussing the connected wedge theorem in de Sitter spacetime. As noted in~\cite{Franken:2023pni}, defining a covariant holographic entanglement entropy prescription in the context of doubled static patch holography leads to several subtleties. In particular, there are situations where the HRT-like formula does not have any solution. A solution to this problem was proposed in~\cite{Franken:2023pni}.\footnote{See \cite{Hao:2024nhd} for related issues in islands computations.} We review this prescription in the following section, in addition to another possibility which is to adapt the maximin prescription of~\cite{Wall:2012uf} to de Sitter spacetime. There are tensions between these three alternative adaptations of the holographic entanglement entropy prescription in AdS/CFT, as studied in detail in the next section. Moreover, we prove in Appendix~\ref{app:proofs} that they are equivalent in the static patch of an observer. This fact is essential for a proof of the connected wedge theorem.

The goal of this paper is not to prove which is the correct prescription. To answer this question one should derive it from a path integral computation. We do not have any particular reason to expect that the maximin method should be equivalent to any holographic entropy prescription. On the other hand, we expect holographic entropy prescriptions in non-AdS spacetimes to involve an extremization process. Indeed, a great number of results in holography, such as entanglement wedge nesting, the boundary causality condition, the inclusion of the causal wedge in the entanglement wedge~\cite{Akers:2016ugt}, as well as the connected wedge theorem~\cite{May:2019yxi,May:2019odp}, heavily rely on the fact that extremal surface is of vanishing (quantum) expansion in all directions.\footnote{In particular, all these properties rely on the classical focusing theorem or restricted quantum focusing conjecture~\cite{Bousso:2015mna,Shahbazi-Moghaddam:2022hbw}, which constrains the evolution of the expansion or quantum expansion along null hypersurfaces.} These conditions emerge from fundamental constraints such as bulk and boundary causality, and we expect them to be satisfied beyond AdS/CFT.

\subsection{Extremal, maximin, and C-extremal surfaces}
\label{app:def}

In AdS/CFT, entropy is given by the area divided by $4G_N$ of the so-called HRT surface~\cite{Hubeny:2007xt}. HRT surfaces are homologous to the associated subsystem and extremal.
\begin{definition}
\label{def:ext}
\textbf{[Extremization]}
The extremal surface of $A$ in $D$ is denoted as $\gamma_e(A;D)$. It is a spacelike codimension-two surface $\gamma$ which is $D$-homologous to $A$ and extremizes the area functional. When there are multiple $D$-homologous surfaces, we choose the one with the smallest area.
\end{definition}
When the bulk region $D$ is not the whole spacetime, the extremization problem may not have a solution. This can happen in the holographic entanglement entropy prescription we are about to define~\cite{Franken:2023pni}. So we must relax our condition of extremality to take into account surfaces that have the smallest or largest area in $D$ without being extremal. In this case, the surface must lie at least partially on the boundary of $D$ otherwise it would be a true extremum of the area functional. Since the usual HRT prescription asks to select the minimal area surface among the extremal surfaces, we only select the surface with the smallest area in $D$. In \cite{Hao:2024nhd}, a similar formula was derived in the context of islands computations in de Sitter space.

\begin{definition}
\label{def:Rext}
\textbf{[Constrained extremization/C-extremization]}
    The C-extremal surface of $A\in\partial D$ in region $D$, denoted as $\gamma_{c}(A;D)$, is a spacelike codimension-two surface $\gamma$, $D$-homologous to $A$ that satisfies one of the two following conditions.
    \begin{enumerate}
        \item It is the extremal surface of $A$.
        \item A positive measure subset\footnote{By a positive measure subset of $\gamma$, we mean that a subset of $\gamma$ of non-vanishing area lies on $\partial D$. This is the negation of the statement that $\gamma\in D \backslash \partial D$ except for a countable number of points $p\in\gamma$.} of $\gamma$ lies on $\partial D$ and $\gamma$ is the minimal area surface in $D$. See Figure~\ref{fig:const_ext} for a schematic example.
    \end{enumerate}
    When there are multiple C-extremal surfaces, we chose the one with the smallest area.
\end{definition}

Here, we do not mean by minimal or maximal area surface in $D$ that the surface is an extremum of the area function, simply that the area takes its minimal or maximal value in $D$. If there exists a C-extremal surface satisfying the second condition of Definition~\ref{def:Rext}, it must be the minimal C-extremal surface.

\begin{figure}[h!]
    \centering
    \includegraphics[width=0.6\linewidth]{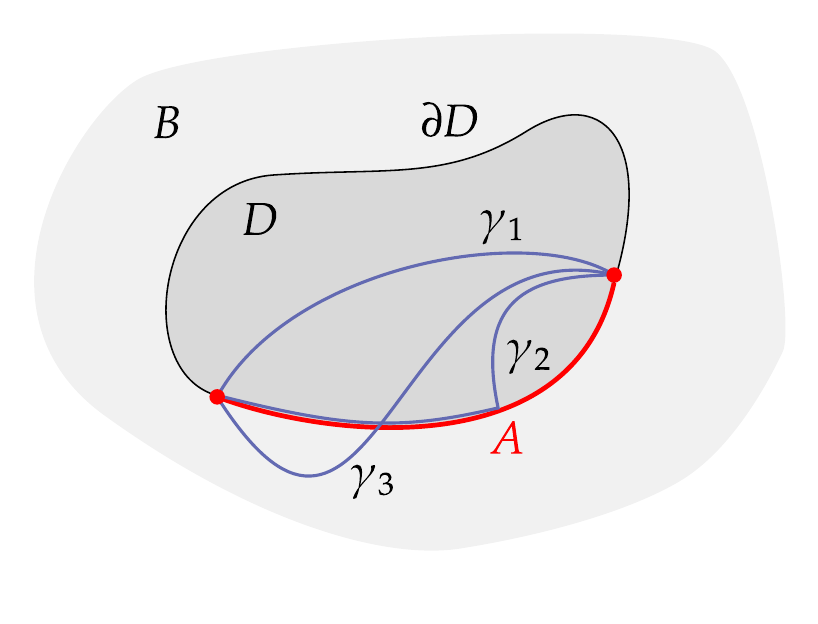}
    \caption{Schematic example of constrained extremization. Consider a region $D$ of the bulk $B$. We look for constrained extremal surfaces (curves here) in $D$ and homologous to a subsystem $A$ of $\partial D$. $\gamma_1$ is extremal and has no positive measure subset on $\partial D$, so it is a constrained extremal curve. $\gamma_2$ is not extremal but has a positive measure subset on $\partial D$ while being the curve of minimal length, so it is a constrained extremal curve. Finally, $\gamma_3$ is extremal in $B$ but is not fully contained in $D$ so it is not a constrained extremal curve in $D$. }
    \label{fig:const_ext}
\end{figure}
In the context of AdS/CFT, the HRT proposal is proven to be equivalent to a maximin procedure~\cite{Wall:2012uf}. Again, we consider the restricted maximin procedure~\cite{Marolf:2019bgj, Grado-White:2020wlb}.
\begin{definition}
    On any achronal slice $\Sigma_I$ such that $A\in\partial\Sigma_I$, which is not necessarily complete, $\gamma_{\min}(A;\Sigma_I)$ is the codimension-two surface $D(\Sigma_I)$-homologous to $A$ with minimal area.
\end{definition}
This definition differs from that of~\cite{Wall:2012uf} by the fact the achronal slices considered here are not complete. This leads to the possibility that $\gamma_{\min}(A;\Sigma)$ is not an extremum with respect to variations along $\Sigma_I$. 

\begin{definition}
\label{def:maximin}
\textbf{[Maximin]}
    Let $\Sigma\vert_{\mathcal{S}}$ be a spacelike slice of some holographic screen $\mathcal{S}$ and $A\in \Sigma\vert_{\mathcal{S}}$ some subsystem. Let $D$ denote the bulk region whose spatial boundary is $\Sigma\vert_{\mathcal{S}}$. The maximin surface associated with $A$, denoted $\gamma_{m}(A;D)$, is the surface $\gamma_{\min}(A;\Sigma')$ which is maximal when varying over all slices $\Sigma'$ such that $\partial\Sigma' = \Sigma\vert_{\mathcal{S}}$. The achronal slice on which $\gamma_{m}(A;D)$ is minimal is denoted $\Sigma_{m}(A;D)$. Additionally, $\gamma_{m}(A;D)$ must be stable in the sense that considering a small variation of $\Sigma_{m}(A;D)$, the new slice must have a local minimum (possibly non-extremal) in the neighborhood of $\gamma_{m}(A;D)$, with no greater area.
\end{definition}

We will not prove the existence of maximin surfaces, as it goes well beyond the scope of this paper. Note that this was proven for certain classes of asymptotically AdS spacetime~\cite{Wall:2012uf, Marolf:2019bgj}. We will always assume that it exists, which does not seem to be an extravagant assumption at least in the simple examples we consider in this paper, that is connected arcs, and possibly unions of them, of spacelike slices of the horizon/stretched horizon in three-dimensional de Sitter spacetime.

\subsection{Proofs of (in)equivalence}
\label{app:proofs}

First, we show a few properties to specify the origin of the inequivalence between the three prescriptions. We then show the equivalence of the three prescriptions in the interior of the holographic screen of an observer.
\begin{lemma}
\label{lem:minimin}
    A codimension-two spacelike surface $\gamma_{\min}(A;D_I)$ that has the smallest area in bulk region $D_I$ and is $D_I$-homologous to $A$ is the minimin surface of $A$ in region $D_I$, and conversely.\footnote{Here minimin is defined by changing the maximization of Definition~\ref{def:maximin} by a minimization.}
\end{lemma}
\begin{proof}
    The smallest area surface $\gamma_{\min}(A;D_I)$ must be $\gamma_{\min}(A;\Sigma_I)$ for any $\Sigma_I$ containing it. The minimal area surface on every other slice cannot be smaller than $\gamma_{\min}(A;D_I)$ otherwise it would itself be the smallest area surface. Hence $\gamma_{\min}(A;D_I)$ is the minimin surface.
\end{proof}

\begin{proposition}[Maximin+Extremality along $\Sigma_m$ $\Rightarrow$ Extremality]
\label{th:ext=maximin}
    A surface $\gamma_I\in D_I$ is the extremal surface of $A$ if it is the maximin surface of $A$ and is extremal along $\Sigma_{m}(A;D_I)$. The converse also holds only if there exists such a maximin surface extremal along $\Sigma_m$.
\end{proposition}
\begin{proof}
    Proof of Theorem~$15$ of~\cite{Wall:2012uf} applies directly here when the maximin surface is extremal along $\Sigma_m$. The main assumptions in~\cite{Wall:2012uf} are that the bulk spacetime is classical, smooth, satisfies the null energy condition, and is globally hyperbolic. All these conditions are considered to be satisfied here, see Section~\ref{sec:assump}. The proof does not depend on the AdS geometry and the only assumption that is nontrivial here is the existence of an extremum on every slice since the slices are not global. To avoid this difficulty we impose explicitly the extremality along $\Sigma_{m}(A;D_I)$. However, even though we assume the maximin surface exists, there might be no maximin surface that is extremal along their $\Sigma_m(A;D_I)$. In such cases, the converse of the proposition is trivially not satisfied, and the extremal surface is a minimax surface.
\end{proof}
An example of the converse of Proposition~\ref{th:ext=maximin} being false is the extremal surface in the exterior region associated with one of the screens. The extremal surface homologous to a screen in the exterior is the bifurcate horizon. However, the stable maximin surface associated with the screen is always the screen itself or the complementary screen. The bifurcate horizon is an unstable maximin surface, but a stable minimax surface.

\begin{theorem}
\label{ineq}
     Maximin, extremization, and C-extremization are inequivalent. Moreover, none of them implies another.
\end{theorem}
\begin{proof}
    Minimin surfaces are not maximin surfaces in general so a C-extremal surface is not always a maximin surface. Conversely, if the maximin surface is not extremal along $\Sigma_m$ and does not lie on $\partial D$, it cannot be C-extremal. Moreover, as we already noted maximin surfaces can fail to be extremal along $\Sigma_{m}$, so being a maximin surface does not imply extremality. Conversely, Proposition~\ref{th:ext=maximin} states that extremality does not imply being a maximin surface. Finally, C-extremality does not imply extremality, and conversely, by Definition~\ref{def:Rext}.
\end{proof}
Theorem~\ref{ineq} implies that the C-extremization and maximin prescriptions introduced in Section~\ref{sec:FPRT} are not appropriate analogs of their AdS counterpart, as they fail to be equivalent. This is not necessarily a problem in the context of the connected wedge theorem since the proof only uses the fact that the extremal surface must be a minimum on at least one slice. This is the case for all C-extremal surfaces and maximin surfaces. The only important issue is that C-extremal surfaces are not always extremal, such that we cannot always apply the focusing theorem. In the rest of this section, we show that the tension between the two prescriptions is absent in the interior of the screen associated with any observer. The analysis of the exterior region is deferred to future work.

For the rest of this section, we denote by $\Sigma$ any achronal slice intersecting an observer's worldline and such that $\partial\Sigma=\Sigma\vert_{\mathcal{S}}$ with $\mathcal{S}$ the holographic screen associated with this observer.
\begin{lemma}
\label{prop:maximax}
    The surface $\Sigma\vert_{\mathcal{S}}$ is not a non-extremal minimal surface on any $\Sigma$.
\end{lemma}
\begin{proof}
    Let $\Sigma$ be an achronal slice with $\partial\Sigma=\Sigma\vert_{\mathcal{S}}$. Let us construct an arbitrary closed codimension-two surface $\tilde\chi \in \Sigma$ arbitrary close to $\Sigma\vert_{\mathcal{S}}$, that is every point of $\tilde\chi$ is infinitesimally close to $\Sigma\vert_{\mathcal{S}}$. Consider a null hypersurface $\mathcal{L}(\Sigma\vert_{\mathcal{S}})$ emanating from $\Sigma\vert_{\mathcal{S}}$ and directed towards the interior of $\mathcal{S}$. From Definition~\ref{def:ap}, $\mathcal{L}(\Sigma\vert_{\mathcal{S}})$ is of non-positive expansion — is a lightsheet — since it is in the interior of the screen. The lightsheets discussed here are depicted in Figure~\ref{fig:proof}. Consider an arbitrary codimension-two surface $\chi\in\mathcal{L}(\Sigma\vert_{\mathcal{S}})$, such that on each generator of the lightsheet $\mathcal{L}(\Sigma\vert_{\mathcal{S}})$, there is exactly one point of $\chi$. Construct the congruence of lightrays $\mathcal{L}(\chi)$ emanating orthogonally from $\chi$ and directed towards the interior of $\mathcal{S}$. This congruence is also of non-positive expansion from Definition~\ref{def:ap}. We define $\tilde\chi = \Sigma \cap \mathcal{L}(\chi)$, which is a closed codimension-two surface on $\Sigma$. $\mathcal{L}(\Sigma\vert_{\mathcal{S}})$ and $\mathcal{L}(\chi)$ being of non-positive expansion, $\tilde\chi$ must have a smaller area than $\Sigma\vert_{\mathcal{S}}$. The procedure could be followed the other way around by starting from $\tilde{\chi}$. We showed that $\Sigma\vert_{\mathcal{S}}$ has an area greater than or equal to that of any of its infinitesimal deformation on $\Sigma$. Hence, if it is a minimal surface on some slice, it must saturate the bound and be an extremum along this slice.
\end{proof}

\begin{figure}[h!]
    \centering
    \includegraphics[width=0.5\linewidth]{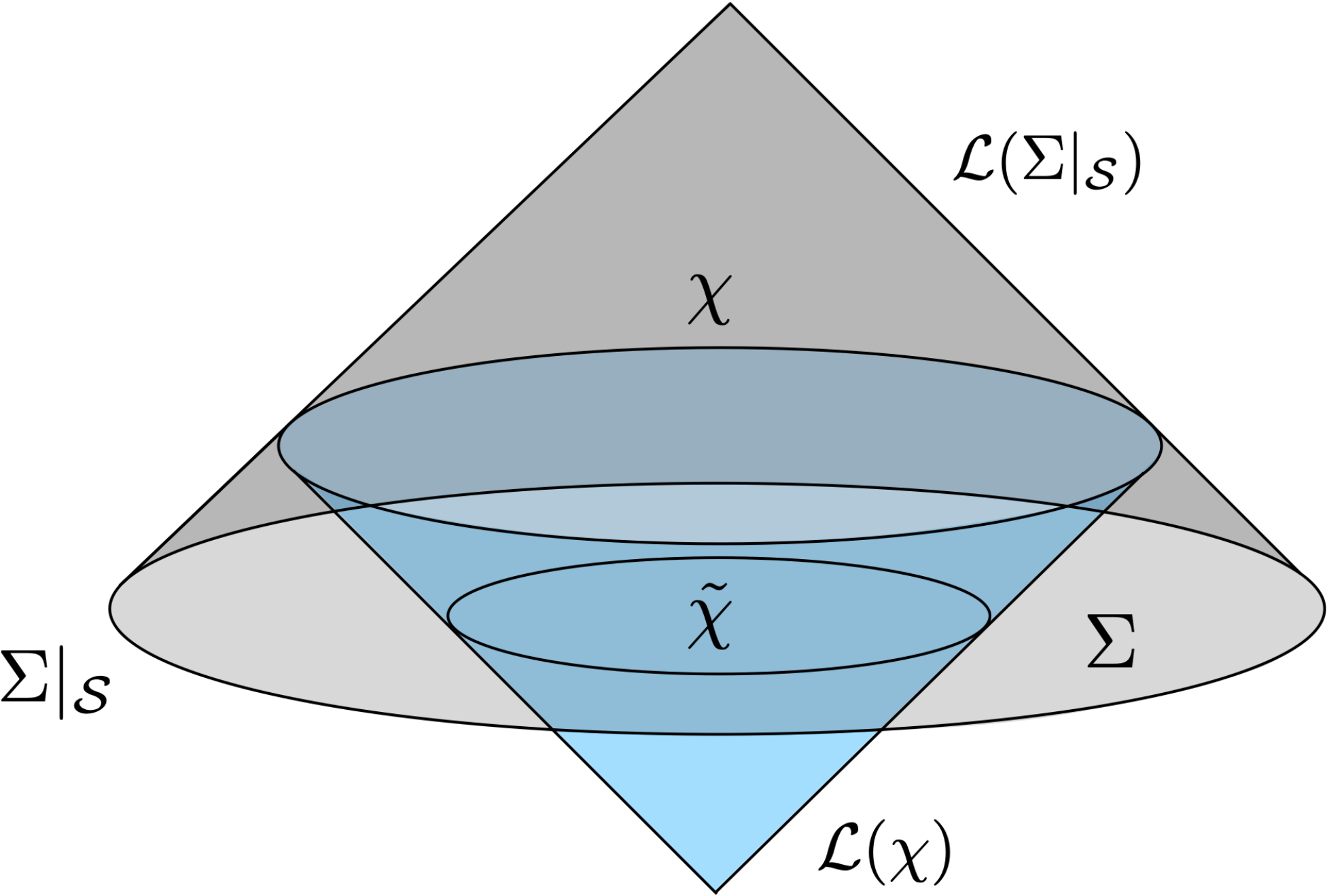}
    \caption{Schematic picture of the proof of Lemma~\ref{prop:maximax}. One constructs an arbitrary closed surface infinitesimally close to $\Sigma\vert_{\mathcal{S}}$ on $\Sigma$ by constructing the lightsheet $\mathcal{L}(\Sigma\vert_{\mathcal{S}})$ (in gray), and then a second lightsheet $\mathcal{L}(\chi)$ (in blue) emanating from an arbitrary closed surface $\chi \in \mathcal{L}(\Sigma\vert_{\mathcal{S}})$ infinitesimally close to $\Sigma\vert_{\mathcal{S}}$. The intersection $\tilde\chi = \mathcal{L}(\chi)\cap\Sigma$ is always smaller than $\Sigma\vert_{\mathcal{S}}$.}
    \label{fig:proof}
\end{figure}

\begin{theorem}
\label{th:patch}
    Maximin surfaces, extremal surfaces, and C-extremal surfaces are equivalent in the interior of the holographic screen of an observer.
\end{theorem}
\begin{proof}
    In this proof, we note $D=D(\Sigma)$ and $A\in\Sigma\vert_{\mathcal{S}}$.
    \begin{itemize}
        \item \textbf{Maximin $\Rightarrow$ Extremal:} From Proposition~\ref{th:ext=maximin}, it is sufficient to show that all maximin surfaces $\gamma_m(A;D)$ are extremal along $\Sigma_{m}(A;D)$. Consider a surface $\gamma$, $D$-homologous to $A$, and infinitesimally close to it. Definition~\ref{def:homol} implies that $\partial A = \partial \gamma$, and $\partial A^c= \partial A$ where $A^c$ is the complement of $A$ on $\Sigma\vert_{\mathcal{S}}$. Hence, $A^c\cup \gamma$ is a closed surface infinitesimally close to $\Sigma\vert_{\mathcal{S}}$, and it lies in on some Cauchy slice $\Sigma$. By Lemma~\ref{prop:maximax},
    \begin{equation}
    \label{eq:ineq}
        \text{area}(A^c\cup \gamma) \leq \text{area}(\Sigma\vert_{\mathcal{S}}) = \text{area}(A^c \cup A),
    \end{equation}
    which implies $\text{area}(\gamma)\leq \text{area}(A)$. Hence, $A$ cannot be a non-extremal minimal surface on any $\Sigma$. An analogous argument applies to a surface $B$ infinitesimally close to $A^c$. $A$ and $A^c$ are therefore not the non-extremal minimal, $D$-homologous to $A$ on any slice $\Sigma$. This extends to any surface $\gamma_A$, $D$-homologous to $A$, that has a positive measure subset on $A$ and or $A^c$. The above arguments show that pushing $\gamma_A \cap \Sigma\vert_{\mathcal{S}}$ in the interior decreases the area, such that $\gamma_A$ is not a non-extremal minimal surface on any slice $\Sigma$.

    A surface that is minimal on some slice $\Sigma$ can only be non-extremal if it has a measurable subsystem on $\mathcal{S}$. However, we showed that such a surface cannot be non-extremal minimum. Hence, any minimal surface on a slice $\Sigma$ must be extremal along $\Sigma$.\footnote{The minimal surface always exists by the extreme value theorem, see proof of Theorem~$9$ in~\cite{Wall:2012uf}, and if it is not on the boundary of $\Sigma$, it is an extremum along $\Sigma$.} This includes the maximin.

    \item \textbf{Extremal $\Rightarrow$ Maximin} is implied by Proposition~\ref{th:ext=maximin} and the above proof that all maximin surfaces are extremal along $\Sigma_m$.
    \item \textbf{Extremal $\Leftrightarrow$ C-extremal:} We proceed by contradiction and assume there exists a surface $\gamma_{\rm min}(A;D)$ $D$-homologous to $A$ with a positive measure subsystem on $\partial D$ while being the smallest area surface in $D$. By Lemma~\ref{lem:minimin}, this surface is a minimin. Any infinitesimal deformation of $\gamma_{\rm min}(A;D)$ along the null direction directed towards the interior of $D$ must decrease the area by Definition~\ref{def:ap}. This is in contradiction with the assumption that $\gamma_{\rm min}(A;D)$ is a minimin.
    \end{itemize}
\end{proof}
These results imply that non-extremal maximin surfaces always arise in the exterior region $D_E$. An example of this is the maximin surface associated with $\Sigma\vert_{\mathcal{S}_L}$ or $\Sigma\vert_{\mathcal{S}_R}$ in the exterior. It is easy to see that this surface is either $\Sigma\vert_{\mathcal{S}_L}$ or $\Sigma\vert_{\mathcal{S}_R}$. In the case where the holographic screens are located on the cosmological horizons of pure de Sitter space, there is an infinite number of maximin surfaces on the horizon, see~\cite{Franken:2023pni}, but the only ones that satisfy the stability conditions are $\Sigma\vert_{\mathcal{S}_L}$ and $\Sigma\vert_{\mathcal{S}_R}$. The only closed extremal surfaces in de Sitter spacetime are the empty surface and the bifurcate horizon, such that the maximin surface is not extremal in this example.\footnote{Except if we consider achronal slices that intersect the horizons at the bifurcate horizon when the holographic screen is the cosmological horizon.} In particular, these surfaces are not extremal along any spacelike Cauchy slice. A similar analysis applies to spatial subsystems of the holographic screens in dS$_3$, which are one-dimensional arcs or a union of them. In particular, the extremal surface $\gamma_e(A;D_E)$ associated with such subsystem does not exist in pure dS$_3$, requiring the use of the maximin or the C-extremization procedure~\cite{Franken:2023pni}.

\bibliographystyle{jhep}
\bibliography{biblio}

\end{document}